\documentclass[11pt,letterpaper]{article}

\usepackage{bm}
\usepackage{nicefrac}
\usepackage[numbers]{natbib}

\usepackage{wrapfig}
\usepackage{paralist}
\usepackage[ruled,vlined]{algorithm2e}
\usepackage[dvipsnames]{xcolor}
\usepackage{booktabs}
\usepackage{tikz}
\usepackage{environ}
\usepackage{pgfplots}
\usepackage{amssymb,amsmath,amsthm}

\pgfplotsset{compat=newest}

\makeatletter
\newsavebox{\measure@tikzpicture}
\NewEnviron{scaletikzpicturetowidth}[1]{%
  \def\tikz@width{#1}%
  \def\tikzscale{1}\begin{lrbox}{\measure@tikzpicture}%
  \BODY
  \end{lrbox}%
  \pgfmathparse{#1/\wd\measure@tikzpicture}%
  \edef\tikzscale{\pgfmathresult}%
  \BODY
}
\makeatother

\usetikzlibrary{shapes}
\usetikzlibrary{arrows}
\usetikzlibrary{calc}

\usepackage{todonotes}
\usepackage{mathtools}
\usepackage{subcaption}
\usepackage{accents}
\usepackage{diagbox}

\newcommand{\NN}[0]{\mathbb{N}}
\newcommand{\RR}[0]{\mathbb{R}}

\newcommand{\D}{\displaystyle}

\newcommand{\TCI}{\textsc{TopCycleIncrease}}

\renewcommand{\vec}[1]{\mathbf{#1}}

\newcommand{\growingmid}{\mathrel{}\middle|\mathrel{}}

\newcommand{\classNP}{\textsf{NP}}
\newcommand{\classPPAD}{\textsf{PPAD}}

\newcommand{\classFIXP}{\textsc{FIXP}}

\newcommand{\SAT}{\textsc{Satisfiability}}
\newcommand{\ThreeDM}{\textsc{3-Dimensional-Matching}}

\newcommand{\sw}{\textsc{Q}}
\newcommand{\rev}{\textsc{Flow}}

\newcommand{\vecf}{\mathbf{f}}
\newcommand{\veca}{\mathbf{a}}
\newcommand{\vecfhat}{\mathbf{\hat{f}}}
\newcommand{\vecpi}{\bm{\pi}}

\newtheorem{theorem}{Theorem}

\newtheorem{example}[theorem]{Example}
\newtheorem{proposition}[theorem]{Proposition}

\newtheorem{remark}[theorem]{Remark}
\newtheorem{corollary}[theorem]{Corollary}

\tikzstyle{vertex}=[circle, minimum size=20pt,inner sep=0pt, draw]
\tikzstyle{smallvertex}=[circle, minimum size=1pt,inner sep=0pt, draw]
\tikzstyle{edge} = [draw,thick,-stealth']

\usepackage[colorlinks=true,breaklinks=true,bookmarks=true,urlcolor=blue,
     citecolor=blue,linkcolor=blue,bookmarksopen=false,draft=false]{hyperref}

\title{Flow Allocation Games\footnote{A previous version of this paper has been published with the title ``Strategic Payments in Financial Networks''. While we receive motivation from financial networks, we study a fundamental game-theoretic approach in the context of classic flow and circulation problems that is not necessarily restricted to financial networks. To reflect this property, we changed the title to ``Flow Allocation Games''.}}

\author{Nils Bertschinger\thanks{Frankfurt Institute of Advanced Studies, Goethe University Frankfurt, Germany. \texttt{bertschinger@fias.uni-frankfurt.de}} \and Martin Hoefer\thanks{Institute for Computer Science, Goethe University Frankfurt, Germany. \texttt{mhoefer@cs.uni-frankfurt.de}} \and Daniel Schmand\thanks{Center for Industrial Mathematics, University of Bremen, Germany. \texttt{schmand@uni-bremen.de}}}

\date{}
\begin{document}
\maketitle
\begin{abstract}
We study a game-theoretic variant of the maximum circulation problem. In a \emph{flow allocation game}, we are given a directed flow network. Each node is a rational agent and can strategically allocate any incoming flow to the outgoing edges. Given the strategy choices of all agents, a maximal circulation that adheres to the chosen allocation strategies evolves in the network. Each agent wants to maximize the amount of flow through her node. Flow allocation games can be used to express strategic incentives of clearing in financial networks.

We provide a cumulative set of results on the existence and computational complexity of pure Nash and strong equilibria, as well as tight bounds on the (strong) prices of anarchy and stability. Our results show an interesting dichotomy: Ranking strategies over individual flow units allow to obtain optimal strong equilibria for many objective functions. In contrast, more intuitive ranking strategies over edges can give rise to unfavorable incentive properties.
\end{abstract}

\section{Introduction.}  

Flows and circulations in networks are a classic problem domain in combinatorial optimization. A network flow is called a feasible circulation in a graph if it maintains flow conservation at all nodes, and some given lower and maybe upper bounds on the flow on edges are fulfilled. Dinitz~\citep{Dinitz70} and Edmonds and Karp~\citep{EdmondsK72} provided strongly polynomial-time algorithms for solving the existence problem of a circulation in a network. \citet{Tardos85} even showed that the minimum-cost circulation problem can be solved in strongly polynomial time. There is a vast number of applications of flow problems, and flows give rise to beautiful and favorable mathematical and algorithmic properties. As a consequence, variants of flow and circulation problems have been investigated for decades.

In this paper, we explore a novel game-theoretic model for circulations in networks. In our model, each node belongs to a player who aims to maximize the flow through the node. We assume that players can strategically allocate the flow entering their node to outgoing edges. Given the strategy choices of all players, a maximal circulation that adheres to the chosen allocation strategies evolves in the network.

Flows and circulations have a broad range of applications to different areas. A fascinating area, where circulations play a key role, is in the analysis of financial networks and systemic risks.
A very popular approach to model financial networks has emerged from the seminal work by \citet{EisenbergN01}. In their model, the financial market can be seen as a directed graph $G=(V,E)$. The node set $V$ corresponds to the set of institutions (or \emph{firms}). The set of directed edges $E$ expresses the debt relations among firms. Each edge  $e=(u,v) \in E$ has a capacity $c_e$ that corresponds to the nominal liability of firm $u$ to firm $v$. In addition, each firm $v$ has a non-negative supply $b_v$ that corresponds to external assets. Given this networked scenario, the goal is to understand the properties of \emph{clearing}, i.e., the resulting payments when firms have to clear their debt and ``pay their bills'' to their creditors. For this clearing task, we view the graph as a flow network, where payments constitute a flow of funds or assets (which we will call \emph{money}). Each node has a non-negative supply and potentially receives additional money over its' incoming edges. It then can use all this money to allocate flow towards its' outgoing edges. As a consequence, money starts to circulate in the network. Eventually, all edges of a node become tight (and all debt is cleared), or the node runs out of funds.

In the majority of the literature, the circulation flow is governed by a static proportional (also called pro-rata) strategy for each node. The node must allocate the entire outgoing flow in proportion to the capacities of its outgoing edges. More recently, interest has emerged in more general, decentralized, and monotone allocation strategies. Notably, in an influential work, \citet{CsokaH18} consider flows based on an arbitrary integral, monotone allocation strategy for each node (which we term \emph{unit ranking} below). This extension leads to a variety of interesting questions for the resulting money circulations.

In this work, our interest lies in the decentralized and, more concretely, incentive and stability properties of the circulation problem. In our model, nodes can individually choose the flow allocation strategy for the outgoing edges. Flow allocation games of this kind have been of interest recently -- in particular, \citet{GuhaKV19} explore games, in which each agent controls the strategies of one or more nodes and wants to maximize the flow routed to an agent-specific sink in the network. Our game is closely related, with the difference that each agent corresponds to a single node, but it might neither be a source nor a sink. Instead, the goal of the node in our game is to maximize the flow circulating through the node. In a financial context, this is equivalent to the natural goal of maximizing the \emph{equity} (total assets minus total liabilities).

Each node can choose as strategy an allocation function that yields for each amount of available flow an assignment of this flow to the outgoing edges. Similar to~\citet{CsokaH18} our interest lies in monotone, ranking-based strategies. Game-theoretic variants of maximum flows, even based on rankings, have been of recent interest in economics. \citet{Fleiner2014} introduced \emph{stable flows}, which have been further developed by \citet{KiralyP13}, \citet{Cseh2013}, and \citet{CsehM19}. In stable flows, each node has an inherent preference ranking over the edges. A node $v$ strives to maximize the amount of flow on its preferred edges. Thus, a stable flow can be seen as a flow that is immune to coalitional deviations of players, i.e., nodes of a non-saturated walk can jointly decide to add flow on this walk. In our work, the ranking is not externally given, but represents a strategic decision of the players.

To the best of our knowledge, flow allocation games studied in this paper have not been addressed before. We provide a cumulative analysis of the properties of equilibria in these games. We focus on pure Nash and strong equilibria. In these equilibria, nodes have no unilateral (pure Nash) or coalitional incentives (strong) to deviate from their chosen strategies. Depending on the set of strategies, the resulting games have different properties. If we assume that strategies are restricted to priority orderings over edges, existence of a pure Nash or a strong equilibrium is not guaranteed and becomes strongly \classNP-hard to decide. Instead, if nodes can assign each unit of available flow in an arbitrary monotone fashion, a strong equilibrium always exists and can be computed in strongly polynomial time. Moreover, this strong equilibrium maximizes the total amount of flow circulating in the network (and, as such, Pareto optimizes the utility of all nodes). In addition, we show that for a diverse set of objective functions, there is a strong equilibrium that is optimal w.r.t.\ this objective (such as, e.g., minimizing the number of nodes or firms that are in default and unable to pay their liabilities).

This interesting technical dichotomy between games with different payment strategies (restricted edge-based vs.\ arbitrary integral and monotone) offers insights into the properties of financial networks. Our results show that a benevolent designer could realize a clearing mechanism with monotone unit-based payment strategies that leads to a socially optimal state, for many different notions of ``social optimum''. It comes with the additional guarantee of giving no coalition of firms an incentive to pay their debts differently. In contrast, if clearing payments are determined in a decentralized fashion resulting in some arbitrary Nash or strong equilibrium, the total amount of flow in the system can deteriorate drastically (and similarly the social quality for many objectives).

Similar problems arise if a centralized mechanism is restricted to edge-based priorities. This can lead to non-existence of pure equilibria in the resulting games. Even if pure equilibria exist, they can be undesirable since, e.g., the total amount of circulating flow can be very small. This shows a marked contrast between centralized and decentralized approaches and highlights how the structure of permissible strategies impacts the structural properties of the resulting flows.

\subsection{Contribution and outline.}
In Section~\ref{sec:model} we introduce our formal model of a flow allocation game. We focus on natural classes of ranking-based payment strategies for the nodes as introduced by~\citet{CsokaH18}. For an \emph{edge-ranking strategy}, a node ranks its outgoing edges and assigns its incoming flow in the order of the ranking\footnote{Cs\'{o}ka and Herings called them priority rules.}. As a superset of strategies, we consider \emph{unit-ranking strategies}, where flow is considered in units. Instead of edges, each node ranks single units of each edge capacity. Edge- and unit-ranking strategies both are classes of \emph{monotone strategies}, where the mapping of the available flow of a node to every outgoing edge is an arbitrary monotone function.

There can be several feasible flows for a given strategy profile. In fact, the feasible flows form a lattice with a partial order based on the total outgoing flow of each node~\cite{CsokaH18}. There is a unique feasible flow that forms the supremum of the lattice -- it pointwise maximizes the outgoing flow to each node (for the given profile of monotone strategies).
We assume that this supremum flow is the clearing flow and defines the utility of each node in the game. After discussing some structural insights on the feasible flows in Section~\ref{sec:clearing}, we show in Section~\ref{sec:topCycle} that the supremum flow can be computed in strongly polynomial time (Proposition~\ref{prop:clearingAlgo}) for profiles of unit-ranking strategies. 

In Section~\ref{sec:coin} we study unit-ranking games, in which all nodes use unit-ranking strategies. Our interest lies in the existence, computational complexity, and social quality of equilibria. We show that in every such game there exists a strategy profile that represents a strong equilibrium, in which no coalition of nodes has an incentive to deviate (Theorem~\ref{thm:coinSPoS1}). Furthermore, there are even strong equilibria that represent a state that maximizes the circulation flow. Moreover, for a large variety of natural notions of social objective functions (such as the sum, the minimum, or the geometric mean of all utilities, the number of fully saturated nodes, etc.) there is a strong equilibrium that maximizes the objective over all possible circulations (Corollary~\ref{cor:coinSPoS1}). 

For simplicity, in the remainder of the paper we then focus on one standard objective in circulation problems, the total amount of flow circulating in the network. A strong equilibrium that maximizes the circulation can be computed in strongly polynomial time (Theorem~\ref{thm:coinSPoS1}). It can even be represented using a number of bits polynomial in the size of the input (i.e., the size of the graph and the size of all numbers in logarithmic encoding), even though every strategy shall rank all single flow units that a node might have available, and their number could be pseudo-polynomial (i.e., linear in the unary encoding size of the input). In contrast, it is strongly \classNP-hard to find a best-response strategy for a single node in a given arbitrary strategy profile of a unit-ranking game (Theorem~\ref{thm:bestRespNPC}).

For worst-case equilibria and the strong price of anarchy, we show that the deterioration of the circulation in a strong equilibrium compared to a max-circulation is tightly characterized by the min-max length of cycles in any max-circulation (Theorem~\ref{thm:spoaD}). This implies that in networks, in which an optimal circulation is composed of small cycles, we see a small inefficiency in strong equilibria. In contrast, the circulation of a worst-case Nash equilibrium, which is stable only against unilateral deviations, can be arbitrarily worse than in an optimum, even in simple games with a constant number of nodes (Proposition~\ref{prop:poaUnbounded}).

In Section~\ref{sec:edge} we study a natural and interesting restriction on the strategies and analyze edge-ranking games, in which all nodes are restricted to edge-ranking strategies. Restricting the strategy space to rankings over edges can have devastating consequences for the existence of equilibria and the amount of circulating flow in an equilibrium. In edge-ranking games, pure Nash and strong equilibria can be absent, and deciding their existence is strongly \classNP-hard (Theorem~\ref{thm:edgeExistNPC}). The same hardness applies for computing an optimal strategy profile, and for computing a pure Nash or strong equilibrium when it is guaranteed to exist. Even the best strong equilibrium can be a factor of $\Omega(n)$ worse than an optimum in terms of the circulating flow (Proposition~\ref{prop:poa}). For pure Nash equilibria, even the best one can be arbitrarily worse than an optimum (Proposition~\ref{prop:posExtrnalInflow}).

We conclude in Section~\ref{sec:conclude} with a summary of the main findings, a discussion of our results, and directions for future work.

\subsection{Related work.}
\label{sec:related}

Flow allocation games are based on circulations in financial network models that emerged from~\cite{EisenbergN01}. Rather than proportional payments,~\citet{CsokaH18} analyze edge- and unit-ranking strategies as well as arbitrary monotone strategies. They analyze the structure of clearing flows and show that they constitute a complete lattice. For completeness, in Section~\ref{sec:model} we recapitulate these findings using our notation.

\paragraph{Flow games.}
Our game-theoretic approach is related to a number of existing game-theoretic models based on flows in networks. In cooperative game theory, there are several notions of flow games based on a directed flow network. Existing variants include games, where edges are players~\cite{DubeyS84,KalaiZ82,KalaiZ82a,GranotG92,DengIN99,BachrachR09}, or each player owns a source-sink pair~\cite{Papadimitriou01,MarkakisS05}. The total value of a coalition $C$ is the profit from a maximum (multi-commodity) flow that can be routed through the network if only the players in $C$ are present. There is a rich set of results on structural characterizations and computability of solutions in the core, as well as other solution concepts for cooperative games. In contrast to our work, these games are non-strategic. We consider each player as a single node with a strategic decision about flow allocation.

More recently, a class of strategic flow games has been proposed in~\citep{GuhaKV19,KupfermanVV17}. There is a capacitated flow network with a set of source nodes. At each source node, a given amount of flow enters the network. Each node of the network is owned by a single player. Each player always owns a designated sink node, as well as one or more additional nodes from the network. A player can choose a flow strategy for each of her nodes. The flow strategy specifies, for every node $v$ and every $x \ge 0$, how an incoming flow of $x$ at $v$ is distributed onto the outgoing edges (if any). Each flow strategy needs to fulfill flow conservation constraints at every node, subject to capacity on the outgoing edges. Each player aims to maximize the incoming flow at its sink node.

For these games there exist a number of $\Sigma_2^p$-completeness results for, e.g., determining the value of a game in a two-player Stackelberg variant, or determining the existence of a pure Nash equilibrium in a multi-player variant. In the latter game, computing a best response can also be \classNP-hard. Our approach is related to these games. However, motivated by financial networks we assume each firm is a single node. The firm optimizes the incoming flow at its node (without it being a designated sink node). We study the computational complexity and social quality of equilibria. Moreover, strategic incentives arise mainly from cycles in the network (see Section~\ref{sec:DAGs} below) -- a condition absent in the existing work on max-flow games~\citep{GuhaKV19,KupfermanVV17} where the network is assumed to be acyclic.

The problem of computing a clearing state for a given strategy profile in our games is closely related to the notion of a stable flow studied in~\citep{Fleiner2014,CsehM19}. In the stable flow problem, each node is equipped with an intrinsic preference order over both incoming and outgoing arcs. The goal is to route as much flow as possible over most preferred arcs. There always exists a stable flow, where no group of agents all can benefit from rerouting the flow along a walk, and such a flow can be computed in polynomial time. The set of stable flows forms a lattice. The model has been extended to an over time setting~\citep{Cseh2013} and to a multi-commodity variant~\citep{KiralyP13}.

\paragraph{Financial networks.} 
We consider issues of strategic choice and computational complexity in flow allocation games. Flow allocation games have a strong connection to clearings in financial networks. There have been works addressing computational complexity of diverse issues in financial networks, such as pricing options with~\citep{AroraBBG11} and without information asymmetry~\citep{BravermanP14}, finding clearing payments with credit default swaps~\citep{SchuldenzuckerSB17}, or estimating the number of defaults when providing a shock in the financial system~\citep{HemenwayK16}. 

In addition, many extensions to the model by Eisenberg and Noe have been proposed in the literature on financial networks. However, even models including cross-holdings of equity~\citep{Suzuki02}, default costs~\citep{RogersV13}, or debt contracts of different seniorities~\citep{Fischer14} follow the idea of the basic approach that all contracts have to be cleared consistently, i.e., clearing payments locally adhere to the rather mechanical clearing rule and constitute a fixed point solution globally.
Indeed, \citet{BaruccaBCDVBC20} have shown that many of the above models can be unified in terms of self-consistent network valuations. A well-known result of such models is the ``robust-yet-fragile'' property exhibited by financial networks, i.e., contagion arises in an all-or-nothing fashion akin to the formation of a giant connected component in random graph models \citep{GaiS10}. This provides important insights into systemic risk and advises the need for macro-prudential regulation.

An extended abstract of this paper has appeared in the proceedings of the ITCS 2020 conference~\cite{BertschingerHS20}. Subsequent to publication of the extended abstract, there has been a significant interest in extending our results and analyzing closely related computational problems in the context of financial networks. \citet{KanellopoulosKZ21} study flow allocation games with credit default swaps (CDS). They obtain a number of \classNP-hardness results for equilibrium existence. Instead of CDSes, \citet{HoeferW22} consider extensions to seniorities (expressed by a class of threshold strategies) along with minimal clearing based on the infimum (rather than the supremum) of feasible flows. For endogeneous seniorities, they extend our main existence result for strong equilibria; for exogenenous seniorities, they observe \classNP-hardness results for equilibrium existence. 

More fundamentally, Ioannidis et al.~\citep{IoannidisKV22ICALP, IoannidisKV23} prove \classFIXP-completeness results for computing a strong approximation of clearing states with debts and CDSes, contrasting \classPPAD-completeness for weaker notions of approximation~\citep{SchuldenzuckerSB17}. Further hardness results for computing optimal clearing states are provided by~\citet{PappW21WINE}. They also study the question of computing the set of banks that default in such networks~\citep{PappW21}. Interestingly, for networks with CDSes, the point in time at which a bank announces to default plays an important role. The same authors analyze incentives for adjusting the network structure to optimize individual funds in the clearing state by deleting an incoming edge or gifting some money to other banks~\citep{PappW20a}. Similar questions of optimally removing debt contracts or bail-outs are analyzed by \citet{KanellopoulosKZ22}. In a related spirit, the effects of \emph{debt swaps} of incoming payment obligations by two banks are analyzed in~\citep{PappW21EC,FroeseHW23}.

In a broader context, strategic and financial aspects of networks have received substantial attention over the last decade. A related, yet orthogonal, body of work considers trading networks~\cite{HatfieldKNOW13,Ostrovsky08}. These models are closely related to two-sided matching under preferences and the study of competitive equilibrium. Rather than circulation effects, agents strive to establish profitable trades with their neighbors. Depending on the model variant, this results in agents matching into pairs, exchanging goods, or establishing upstream/downstream relations with suppliers and customers. The analysis of these networks usually addresses similar issues as the ones we consider here, such as existence, structure, and computational complexity of equilibria. Equilibrium computation in trading networks has recently been studied in, e.g.,~\cite{CandoganEV21,FleinerJST23}. These works also contain a good overview of related work and pointers into the existing literature.

\section{Flow allocation games.}
\label{sec:model}

\subsection{Network model, monotone strategies, and utilities.}

\paragraph{Network model.}
In a \emph{flow allocation game} $\Gamma=(G,(b_v)_{v \in V},(c_e)_{e \in E})$ we are given a graph $G=(V,E)$ with a node set $V$ and a set of directed edges $E$. Each node $v \in V$ corresponds to a \emph{player} and has a \emph{fixed supply} $b_v \ge 0$. The \emph{capacity} $c_e \ge 0$ of an edge $e = (u,v)$ is the maximum amount of flow that $u$ can forward to $v$. In terms of financial networks, $b_v$ is the amount of external assets of $v$, and $e$ represents a liability of value $c_e$ of firm $u$ to firm $v$. We follow standard notation in graph theory and denote by $E^+(v) = \{(v,u) \in E\}$ and by $E^-(v) = \{ (u,v) \in E\}$ the set of outgoing and incoming edges of $v \in V$, respectively. The \emph{saturating output} $c^+_v$ of a node $v$ is the maximum amount of flow $v$ can send to other players, specified by the weighted outdegree
\[ c^+_v = \sum_{e \in E^+(v)} c_e.\]
We strive to analyze issues of computational complexity. As such, we will assume that all numbers in the input, i.e., all $b_v$ and $c_e$, are integer numbers in binary encoding.

\paragraph{Allocation strategies.}
We analyze allocation strategies in flow allocation games. Each player $v \in V$ strategically allocates its \emph{total supply}, i.e., the fixed supply $b_v$ plus the total amount of incoming flow, to the outgoing edges $E^+(v)$. More formally, each player $v \in V$ chooses as a \emph{strategy} a parametrized flow allocation function $a_e: \mathbb{R}_{\geq 0} \rightarrow \mathbb{R}_{\geq 0}$ for every outgoing edge $e \in E^+(v)$. This function specifies an amount of flow $a_e(y)$ that is forwarded to the edge $e$, for every $y \in \mathbb{R}_{\ge 0}$. Here, $y$ represents a possible value of total supply. Intuitively, the strategy specifies for every possible value $y \ge 0$ of total flow that $v$ might have available, how $v$ will allocate this flow to the outgoing edges. The strategy $\veca_v = (a_e)_{e \in E^+(v)}$ of player $v \in V$ must satisfy for every $y \ge 0$ and $e \in E^+(v)$
\begin{eqnarray}&0 \le a_e(y) \le c_e\;, & \hspace{.3cm} \text{(capacity constraint)}\\[0.2cm]
                &\D \sum_{e \in E^+(v)} a_e(y) \leq y\;, & \hspace{.3cm} \text{(weak flow conservation constraint)}\\[0.2cm]
                &\D \sum_{e \in E^+(v)} a_e(y) = \min\{c^+_v, y\}\;. & \hspace{.3cm} \text{(no-fraud constraint)} \label{eq:nofraud}
\end{eqnarray}
The capacity constraint ensures that no edge is overused, and the weak-flow conservation constraint ensures that $v$ cannot generate additional flow. The no-fraud constraint ensures that each player forwards as much of its total supply as possible. No-fraud strategies are desirable in the application of financial networks, since they do not allow financial firms to hide or malversate assets while having unpaid debt (hence the name ``no-fraud'')\footnote{For example, \citet{EisenbergN01} restrict the choice of each player to a specific no-fraud strategy with \emph{pro-rata} payments: Each node $v$ distributes its total supply in proportion to the edge capacities. Formally, for every edge $e \in E^+(v)$ the strategy is fixed to $a_e(y) = \min\left\{ c_e , y \cdot \frac{c_e}{c^+_v} \right\}$.}. We assume this property in the strategies for simplicity. Note that even if we drop the no-fraud condition as a constraint, then in flow allocation games (with utilities resulting from monotone strategies defined below) it turns out that every player always has a best response that satisfies the no-fraud condition (see Section~\ref{sec:clearing} below).

A \emph{flow} is a vector of edge valuations $\vecf = (f_e)_{e \in E}$. Given some flow $\vec f$, we slightly abuse notation and denote the total supply of $v$ by $f_v = b_v + \sum_{e \in E^-(v)}f_e$. For a \emph{strategy profile} $\veca = (\veca_e)_{e \in E}$, a flow is called \emph{feasible} if for all edges $(v,w) \in E$ it holds
\begin{equation}
f_{(v,w)} = a_{(v,w)}(f_v)\;. \hspace{1cm} \text{(fixed point constraint)} 
\end{equation}
A node is called \emph{fully saturated} if the feasible flow saturates all outgoing edges, i.e., $f_e = c_e$ for all $e \in E^+(v)$, or equivalently $f_v = c^+_v$. Given some feasible flow $\vec f$ for a strategy profile $\vec a$, the \emph{utility} of player $v$ is defined by $u_v(\vec a, \vec f) = \sum_{e \in E^+(v)} f_e$, i.e., $v$'s goal is to choose a strategy to maximize the total outgoing flow. Capacity and weak flow conservation constraints are not sufficient for a consistent definition of utility. In particular, even if a strategy profile satisfies capacity and weak flow conservation constraints, it may not allow a feasible flow. 

\begin{example} \rm
\label{example_noFeasibleFlow}
Consider a graph depicted in Fig.~\ref{fig:noFlow} with three nodes $V=\{v_1,v_2,v_3\}$ and edges $(v_1,v_2)$, $(v_2,v_1)$ and $(v_1,v_3)$. Let $c_e = 2$ for all edges $e \in E$. For the strategy profile $\vec a$, we assume 
\begin{align*}
	v_2 \text{ plays } &a_{(v_2,v_1)}(y) = y\\[0.2cm]
	v_1 \text{ plays } &a_{(v_1,v_2)}(y) = \begin{cases} y & \text{ for } y \le 1\\
		0  & \text{ for } y > 1\end{cases}\\
	&a_{(v_1,v_3)}(y) = y - a_{(v_1,v_2)}(y) \text{ for } y \ge 0.
\end{align*}
Informally, $v_1$ forwards the total supply to $v_2$ if it is at most 1. Otherwise, $v_1$ forwards all supply to $v_3$. The fixed supply is $b_{v_1} = 1$ and 0 for the other nodes.
\begin{figure}[ht]
	\begin{center}
		\begin{scaletikzpicturetowidth}{0.5\textwidth}
			\begin{tikzpicture}[scale=\tikzscale,auto,swap]
				\node[vertex,label=above:\fbox{1}](v1) at (0,0){$v_1$};
				\node[vertex](v2) at (3,0){$v_2$};
				\node[vertex](v3) at (1.5,1.5){$v_3$};
				\path[edge] (v1) edge [bend left = 20,above] node{$2$} (v2);
				\path[edge] (v2) edge [bend left = 20,below] node{$2$} (v1);
				\path[edge] (v1) edge [left] node {$2$}(v3);
			\end{tikzpicture}
		\end{scaletikzpicturetowidth}
	\end{center}
	\caption{The graph used in Example \ref{example_noFeasibleFlow}.}
	\label{fig:noFlow}
\end{figure}

There is no feasible flow for strategy profile $\vec a$. Suppose the total supply of $v_1$ is 1, then $v_1$ routes all supply to $v_2$, who forwards it back to $v_1$. The total supply of $v_1$ must be at least 2, a contradiction. Suppose the total supply of $v_1$ is more than 1, then $v_1$ forwards all supply to $v_3$, so $v_1$ only has the fixed supply. The total supply of $v_1$ must be 1, a contradiction. \hfill $\blacksquare$
\end{example}

The main problem with feasible flows in the example is that the strategies are not \emph{monotone}. 

\paragraph*{Monotonicity.}
Monotonicity is a natural condition for payment strategies. In a flow allocation game with \emph{monotone strategies}, each player $v \in V$ forwards the flow in a monotone fashion. Monotone strategies are characterized by capacity and weak-flow conservation constraints, and, for every $y,y' \in \RR_{\geq 0}$ with $y \geq y'$
\begin{equation}
\label{eq:monotone}
a_e(y) \geq a_e(y')\;. \hspace{1cm} \text{(monotonicity constraint)} 
\end{equation}

Monotone strategies have been proposed and studied before by \citet{CsokaH18}. In the following theorem, we recapitulate their main structural insight on feasible flows. Consider a flow-allocation game and a strategy profile $\veca$ of monotone strategies. Let $\mathcal{F}$ be the set of feasible flows for $\veca$. We observe that $\mathcal{F}$ is non-empty, i.e., a feasible flow always exists. Moreover, $(\mathcal{F}, \le)$ forms a lattice with the coordinate-wise comparison. Formally, $\vecf \le \vecf'$ iff $f_e \le f'_e$ for all $e \in E$; and $\vecf < \vecf'$ iff $f_e \le f'_e$ for all $e$ and $f_e < f'_e$ for at least one edge $e$. 

\begin{theorem}[\cite{CsokaH18}]
\label{thm:clearingLattice}
For every strategy profile $\veca$ in a flow-allocation game with monotone strategies, the pair $(\mathcal{F}, \le)$ is a non-empty complete lattice.
\end{theorem}

Csoka and Herings state the result only for \emph{integral} monotone strategy profiles $\veca$. We here reiterate the proof to show that it works for all monotone ones. This theorem motivates the study of flow allocation games with monotone strategies. It follows using the Knaster-Tarski theorem.

\begin{theorem}[\cite{tarski1955lattice}]
\label{thm:KnasterTarski}
Let $(L,\leq)$ be any complete lattice. Suppose $g : L \to L$ is order-preserving, i.e., for all $x,y \in L$ we have that $x \leq y$ implies $g(x) \leq g(y)$. Then the set of all fixed points of $g$ is a non-empty, complete lattice with respect to $\leq$.
\end{theorem}

\begin{proof}[Proof of Theorem \ref{thm:clearingLattice}.]
Consider $F = \left\{ \vecf \growingmid 0 \le f_e \le c_e, \hspace{0.15cm} \forall e \in E \right\}$, a compact superset of all possible flow vectors. Obviously, $(F,\leq)$ forms a complete lattice with the coordinate-wise comparison defined above. For a given strategy profile $\veca$, the map $g : F \to F$ with
\begin{equation}
\label{eq:fixedPoint}
g(\vecf)_{(u,v)} = a_{(u,v)}\left(\sum_{e' \in E^-(u)}f_{e'} + b_u\right) \hspace{1cm} \text{ for every } (u,v) \in E
\end{equation}
is an order-preserving function for every edge $e \in E$, since the strategies are monotone. Obviously, the set of feasible flows $\mathcal{F}$ is the set of fixed-points of $g$. The result follows by applying the Knaster-Tarski theorem.
\end{proof}

\paragraph{Clearing states and utilities.}
We consider feasible flows that arise due to the strategic flow allocation decisions. We determine the utility $u_v(\veca, \vecf)$ of a player $v$ by using a feasible flow $\vecf$ for the given strategy profile $\vec a$. We focus on monotone strategies in order to guarantee the existence of at least one feasible flow. This is a necessary condition to make the game well-defined. In many cases, for a fixed strategy profile $\veca$, there is a unique feasible flow $\vecf$. However, even in very special cases, there might be infinitely many feasible flows for the same strategy profile\footnote{For example, consider a simple cycle with two nodes $v$ and $w$. The capacities $c_{(v,w)} = c_{(w,v)} = 1$, the fixed supplies $b_v = b_w = 0$, and the strategies $a_{(v,w)}(y) = a_{(w,v)}(y) = y$. Every flow with $f_{(v,w)} = f_{(w,v)} \in [0,1]$ is a feasible flow.} $\veca$.

Based on these properties, and similarly to the vast majority of the literature, we concentrate on the supremum of the lattice. For every profile $\veca$, we focus on the unique feasible flow that maximizes the total flow in the network. We denote this feasible flow by $\vecfhat$ and call it the \emph{clearing state} of strategy profile $\veca$. The clearing state $\vecfhat$ determines the utility of every player in $\veca$, i.e., 
\[ u_v(\veca) = u_v(\veca, \vecfhat) =  \sum_{e \in E^+(v)} \hat{f}_e \enspace.\]
Clearly, the clearing state $\vecfhat$ and the resulting utilities significantly depend on the strategy choices of the nodes in the profile $\veca$. 

\subsection{Ranking-based strategies.}
In this paper we are interested in classes of intuitive, expressive, and meaningful monotone strategies. We concentrate on strategies that can be derived via \emph{rankings}~\cite{CsokaH18}.

\paragraph{Unit-ranking strategies.}
In unit-ranking games, we rely on integrality of all values for $c_e$ and $b_v$, and choose the strategies such that the feasible flows will be integral. We define the parametrized flow functions $a_e(y)$ for every outgoing edge $e \in E^+(v)$ on the non-negative integer numbers $a_e(y): \NN_0 \rightarrow \NN_0$. Note that Theorem~\ref{thm:clearingLattice} can also be shown for $a_e(y): \NN_0 \rightarrow \NN_0$. The proof is analogous, where we define $F$ on integrals and replace \emph{compact} by \emph{finite}. This ensures that unit-ranking strategies are well-defined in our model. We interpret the flow as being discretized into unsplittable ``units'' or ``particles'' of size 1.

Unit-ranking strategies are relevant for the application of flow-allocation games to financial networks. Usually, currencies have some smallest indivisible amount of money. In this way, unit-ranking strategies provide a rich and powerful class of strategies that can be used to express payments strategies in this context.

\paragraph{Edge-ranking strategies.}
In a flow-allocation game with edge-ranking strategies, each player $v \in V$ forwards its total supply according to a strict and total order over $E^+(v)$, which we represent by a permutation $\pi_v = (e_1,e_2,\ldots)$. Player $v$ first allocates the maximum possible flow to edge $e_1 = \pi_v(1)$, then $e_2 = \pi_v(2)$, etc.\ until all edges are at their capacity or $v$ has no supply left. Formally, $a_{e_i}(y) = \min\{ c_{e_i},\max\{ 0, y - \sum_{j < i} c_{e_j} \}\}$. The \emph{edge-ranking strategy}\footnote{These strategies have also been termed \emph{singleton liability priority lists} in~\citep{IoannidisKV23}.} of $v$ is fully described by the ranking $\pi_v$, hence we denote a strategy profile in edge-ranking games by $\vecpi = (\pi_v)_{v \in V}$.\\

Edge-ranking strategies are a special case of unit-ranking strategies in the sense that every edge-ranking strategy can be written as a unit-ranking strategy. Maybe counterintuitively, every unit-ranking game is also a special edge-ranking game -- replacing each edge $e$ with capacity $c_e$ many multi-edges of unit capacity expands a unit-ranking game into an equivalent edge-ranking game. There is a one-to-one correspondence between unit-ranking strategies in the original game and edge-ranking strategies in the expanded game. Intuitively, for a unit-ranking strategy in the original game, a player $v$ assigns the first particle of flow to the multi-edge $\pi_v(1)$, the second particle to $\pi_v(2)$, etc.\ in the expanded edge-ranking game until all outgoing edges are saturated or $v$ runs out of supply. The expansion of the game implies a pseudo-polynomial blowup in representation size. Nevertheless, the structural equivalence turns out to be very useful for characterizing and analyzing feasible flows and equilibria in unit-ranking games.

Note that the trivial representation of a unit-ranking strategy might be pseudo-polynomial in the size of the original non-expanded game, since edge capacities are given in binary encoding. We will address this issue briefly in Theorem~\ref{thm:coinSPoS1} when we discuss polynomial-time computation of equilibria. It turns out that there always exist equilibria with a representation that is polynomial in the input size of the game.

\subsection{Equilibria and Social Quality}
\paragraph{Equilibrium concepts.} 
We study pure Nash and strong equilibria of flow allocation games. A \emph{(pure) Nash equilibrium} in a flow allocation game is a strategy profile $\veca$ such that no player $v$ has an incentive to unilaterally deviate from the strategy $\veca_v$. More formally, in a pure Nash equilibrium $\veca$ we have $u_v(\veca) \ge u_v(\veca'_v, \veca_{-v})$ for every player $v \in V$ and every strategy $\veca'_v$. Here $\veca_{-v}$ denotes the reduced profile composed of all entries of $\veca$ except the entries for player $v$.

For the definition of a \emph{strong equilibrium}, we first define the notion of a profitable deviation of a coalition. A coalition $C \subseteq V$ of nodes has a \emph{profitable deviation} $\veca'_{C} = (\veca'_v)_{v \in C}$ if upon joint deviation of $C$ to $\veca'_C$, the resulting utility in the new profile $(\veca'_{C}, \veca_{-C})$ is strictly better for every player in $C$, i.e., $u_v(\veca'_{C}, \veca_{-C}) > u_v(\veca)$ for every $v \in C$. Here $\veca_{-C}$ denotes the reduced profile composed of all entries of $\veca$ except the ones for players $v \in C$. A strategy profile $\veca$ is a \emph{strong equilibrium} if no coalition $C \subseteq V$ has any profitable deviation. Thus, by definition, a strong equilibrium is a Nash equilibrium.

In general, pure Nash or strong equilibria might not exist in a flow allocation game. If they are guaranteed to exist, they might not be unique. 

\paragraph{Prices of anarchy and stability.}
In addition to existence and computational complexity, we also quantify the performance of a feasible flow in equilibrium in terms of natural notions of social quality. For any non-negative objective function $\sw(\veca)$ measuring the quality of a strategy profile, we rely on standard notions of \emph{price of anarchy} and \emph{price of stability} to relate the quality in equilibrium to the one that could be obtained in a strategy profile that maximizes $\sw$.

The \emph{price of anarchy} for an equilibrium concept in a game $\Gamma$ is given by the ratio 
\begin{equation}
	\label{eq:poa}
	\max_{\veca \in \mathcal{A}_{\text{Eq}}(\Gamma)} \;\; \frac{\sw(\veca^*)}{\sw(\veca)} \quad = \quad \frac{\sw(\veca^*)}{\min_{\veca \in \mathcal{A}_{\text{Eq}}(\Gamma)} \sw(\veca)}\enspace.
\end{equation}
Here $\mathcal{A}_{\text{Eq}}(\Gamma)$ is a set of equilibria (e.g., the set of all Nash equilibria, or the set of all strong equilibria) of the game $\Gamma$, and $\veca^*$ is a strategy profile maximzing $\sw$. The price of anarchy for a class of games is the largest price of anarchy of any game in the class. The \emph{price of stability} for an equilibrium concept in a game $\Gamma$ is defined by replacing $\max$ with $\min$ and vice versa in~\eqref{eq:poa}. The price of stability for a class of games is the largest price of stability in any game in the class. Note that both prices of anarchy and stability are at least 1.

Intuitively, an upper bound of $\rho$ on the price of anarchy implies that \emph{every equilibrium} has a quality of at least $\sw(\veca^*)/\rho$ in \emph{every game} of the class. A lower bound of $\rho$ implies that for \emph{some game} there is \emph{some equilibrium} in that game with quality at most $\sw(\veca^*)/\rho$. Similarly, an upper bound of $\rho$ on the price of stability implies that \emph{at least one equilibrium} has a quality of at least $\sw(\veca^*)/\rho$ in \emph{every game} of the class. A lower bound of $\rho$ implies that for \emph{some game} it holds that \emph{every equilibrium} in that game with quality at most $\sw(\veca^*)/\rho$.

Our main result for unit-ranking games shows the existence of optimal strong equilibria. More in detail, the price of stability is 1 for pure Nash and strong equilibria for every quality function from a large class of so-called \emph{flow-monotone functions} $\sw$ (for details see Corollary~\ref{cor:coinSPoS1} below). 

For the remaining results on prices of anarchy and stability in this paper, we concentrate on the total amount of flow, which is a standard objective for circulation and flow allocation problems. Formally, the \emph{total (amount of) flow} for a strategy profile $\veca$ is
\begin{equation}
	\label{eq:flowSW}
	\rev(\veca) \; = \;  \sum_{e \in E} \hat{f}_e = \sum_{v \in V} \sum_{e \in E^+(v)} \hat{f}_e = \sum_{v \in V} u_v(\veca) \enspace,
\end{equation}
since the clearing state $\vecfhat$ determines the utilities for all players. Hence, $\rev(\veca)$ also represents the \emph{utilitarian welfare}.

\subsection{Games on DAGs}
\label{sec:DAGs}

Flow allocation games are designed to analyze incentives in networks with circulation flows. In contrast to previous work~\cite{GuhaKV19, KupfermanVV17}, incentives in our games are inherently connected to cycles. To see this, we briefly discuss games on directed acyclic graphs (DAGs). 
	
\begin{proposition}
	In a flow allocation game on a DAG $G$, every strategy profile is a strong equilibrium.
\end{proposition}
	
\begin{proof}
We prove the statement inductively. Consider an arbitrary profile $\veca$. We first discuss a natural algorithm to construct a feasible flow. A DAG contains a node $v$ without incoming edges. Clearly, the incoming assets of $v$ are fixed to $b_v$. Thus, every strategy $\veca_v$ allocates an amount of $\min\{c_v^+,b_v\}$ to the outgoing edges. This is independent of the behavior of $v$ or any other player, so $\veca_v$ represents a best response. The algorithm then constructs an equivalent network for the remaining players by (1) removing $v$ and all its outgoing edges, and (2) increasing $b_w$ by $a_{(v,w)}(b_v)$, for every $(v,w) \in E^+(v)$. The new network is again a DAG and contains a node without incoming edges. In this way, the algorithm proceeds until the remaining network contains no edges.
	
To show that $\veca$ is a strong equilibrium, consider any coalition $C$ and the first player $v_1 \in C$ that was processed by our algorithm above. $v_1$ and all players being processed after $v_1$ cannot increase the incoming assets of $v_1$ by changing their strategies. As such, $v_1$ has no incentive to deviate from $\veca_{v_1}$. $C$ has no profitable deviation. Therefore, $\veca$ is a strong equilibrium.
\end{proof}
	
Our arguments do not rely on monotonicity -- the algorithm constructs a feasible flow even for non-monotone strategies $\veca$. Moreover, by the same induction, for any given strategy profile $\veca$ the feasible flow is unique. The proposition holds even for flow allocation games with non-monotone strategies. Since \emph{every} strategy profile is a strong equilibrium, it is stable against arbitrary deviations of coalitions. Hence, the result continues to hold even when we restrict to games on DAGs with monotone unit- or edge-ranking strategies.

\section{Properties of feasible flows.}

We observe a useful \emph{circulation representation} of feasible flows in flow allocation games and some preliminary results that will be used in the subsequent sections. For edge-ranking games, the circulation representation can be used to describe a polynomial time algorithm that computes the clearing state $\vecfhat$ in polynomial time.

\subsection{Preliminaries on feasible flows and clearing states.}
\label{sec:clearing}

\paragraph{Circulation structure.}
Given a game $\Gamma$ on $G$ with strategies $\veca$, we build a \emph{circulation network} $G'$ with strategies $\veca'$ and consider an extended game in $G'$ as follows. We add to $G$ an auxiliary node $s$. For every $v \in V$, we add an auxiliary edge $(v,s)$ with capacity $c_{(v,s)} = \infty$. For every $v \in V$ with $b_v > 0$ we add an auxiliary edge $(s,v)$ with $c_{(s,v)} = b_v$, and choose $b_v' = 0$ in the new game. In this way, the supply of $v$ becomes an incoming flow to $v$ on edge $(s,v)$. The strategy $\veca'_s$ of player $s$ is arbitrary. For every strategy profile $\veca$ of the original game, we modify $\veca_v$ such that flow not forwarded by some player $v$ will now be a flow on edge $(v,s)$ under $\veca'_v$.

\begin{proposition}
\label{prop:sFull}
For every feasible flow $\vecf$ for a strategy profile $\veca$, the corresponding flow for the modified strategy profile $\veca'$ in the extended game $G'$ can be decomposed and represented as a circulation. The auxiliary node $s$ has an incoming flow of $\sum_{v \in V} b_v$, and all auxiliary edges $(s,v)$ are saturated.
\end{proposition}

\begin{proof}
The proposition is a simple consequence of fixed point constraint and no-fraud constraint. Non-forwarded flow at firm $v \in V$ in the original game exists only if $v$ saturates all outgoing edges
\[ \sum_{e \in E^+(v)} a_e(f_v) = \min\{f_v, c^+_v\} \enspace.\]
Moreover, the total sum of non-forwarded flow in the original game is exactly the sum of all fixed supply:
\begin{align*}
\sum_{v \in V} b_v \quad &= \quad \sum_{v \in V} b_v + \sum_{v \in V} \sum_{e = (u,v) \in E^-(v)} a_e(f_u) - \sum_{v \in V} \sum_{e=(u,v) \in E^-(v)} a_e(f_u) \\
&= \quad \sum_{v \in V} \left(f_v - \sum_{e \in E^+(v)} a_e(f_v)\right)\quad=\quad \sum_{v \in V} \max\{0, f_v - c^+_v\}\;.
\end{align*}
The non-forwarded flow in the original game gets routed to the auxiliary node $s$ in the extended game. This constitutes the incoming flow of $s$, i.e., $f_s = \sum_{v \in V} \max\{0, f_v - c^+_v\} = \sum_{v \in V} b_v$, and all auxiliary edges $(s,v)$ are saturated. Overall, by routing the non-forwarded flow in the original game to the auxiliary node $s$ in the extended game, we obtain exact flow conservation at every node. As such, the flow can be decomposed and represented as a circulation.
\end{proof}

\paragraph{Fully saturated nodes.} In flow allocation games $\Gamma$, it turns out that the clearing state $\vecfhat$ is fixed as long as all nodes that are not fully saturated stick to their strategies. Since every strategy satisfies capacity and no-fraud constraints, the forwarded flow of fully saturated nodes remains the same if they have the same total supply, and vice versa. Consequently, strategies of fully saturated nodes have no impact on the clearing state $\vecfhat$. For any fully saturated node $v$, every strategy constitutes a best response.

\begin{proposition}
\label{prop:solventUnique}
For a given flow allocation game, consider any profile $\veca$, the corresponding clearing state $\vecfhat$, and any fully saturated node $v$ with $\hat{f}_v \ge c_v^+$. Every strategy $\veca'_v$ is a best response for $v$ against the other strategies $\veca_{-v}$ and results in the same clearing state $\vecfhat$.
\end{proposition}

\begin{proof}
Firm $v$ is fully saturated under $\vecfhat$, thus $\hat{f}_v \geq \sum_{e \in E^+(v)}c_e$ and $\hat{f}_e = c_e$ for all $e \in E^+(v)$. Consider a deviation $\veca'_v$, the resulting profile $\veca' = (\veca'_v,\veca_{-v})$. It suffices to show that $\vecfhat$ is feasible under $\veca'$. Note that capacity, weak flow conservation and no-fraud constraints ensure that $a'_e(\hat{f}_v) = c_e$ for all $e \in E^+(v)$, which immediately implies the feasibility of $\vecfhat$ under $\veca'$.
\end{proof}

\paragraph{No-Fraud Property.} 

We observe that violating the no-fraud constraint is never in the interest of any player in any game.

\begin{proposition}
	\label{prop:no-fraudNice}
	Suppose $\veca$ is a strategy profile of monotone strategies that do not necessarily fulfill the no-fraud constraint~\eqref{eq:nofraud}, and strategy $\veca_v$ of player $v$ is such that for a value $y$ of total supply we have $\sum_{e \in E^+(v)} a_e(y) < \min\{c^+_v,y\}$. There is a no-fraud strategy $\veca'_v$ such that $u_w(\veca) \le u_w(\veca'_v,\veca_{-v})$ for every player $w \in V$.
\end{proposition}
\begin{proof}
Consider any no-fraud strategy $\veca_v'$ arising from $\veca_v$ by increasing the functions $a_e$ arbitrarily such that Condition~\eqref{eq:nofraud} holds. Consider the clearing states $\vecfhat$ for $\veca$ and $\vecfhat'$ for $(\veca_v', \veca_{-v})$. For the same total supply $y$ of $v$ we know $a_e(y) \le a'_e(y)$ for every $e \in E^+(v)$. Now consider $\vecfhat$. If $\vecfhat$ is a feasible flow for $(\veca_v', \veca_{-v})$, then $\vecfhat' \ge \vecfhat$, so the utility of every player in $(\veca_v', \veca_{-v})$ is weakly improved and we are done. 

Otherwise, $\vecfhat$ is not a feasible flow for $(\veca_v', \veca_{-v})$. We apply the straightforward monotone fixed-point iteration based on the map $g$ in Equation~\eqref{eq:fixedPoint}. Since the space $F$ of possible flow vectors is compact, the iteration converges to a feasible flow for $(\veca_v', \veca_{-v})$ that is coordinate-wise at least $\vecfhat$. Hence $\vecfhat' \ge \vecfhat$, so the utility of every player in $(\veca_v', \veca_{-v})$ is weakly improved.
\end{proof}

\subsection{Structure and computation of feasible flows in games with edge-ranking strategies.}
\label{sec:topCycle}

For edge-ranking games, we provide a more detailed analysis of the structure of feasible flows in a strategy profile $\vecpi$. We observed above that every feasible flow $\vecf \in \mathcal{F}$ is in a one-to-one correspondence to a circulation in the circulation network. Here, we observe that the circulation follows a partial order. We use this structural insight to show that the clearing state $\vecfhat$ can be computed in polynomial time.

Before we describe the algorithm, we will define the notion of an \emph{active edge}. For a given strategy profile $\vecpi$, consider the circulation network $G'=(V',E')$ with auxiliary node $s$ and the corresponding auxiliary edges described in the previous section. We also extend $\vecpi$ to $\vecpi'$ and define the auxiliary edge $(v,s)$ to be the least preferred edge by node $v$. For completeness, we treat $s$ as a node with a fixed strategy $\pi'_s$ over its outgoing auxiliary edges. Due to Proposition~\ref{prop:sFull}, $s$ exactly saturates all outgoing edges in every feasible flow. As such, the strategy $\pi'_s$ has no impact on the set of feasible flows. 

Let $\vecf$ be a feasible flow in a circulation network $G'$. In games with ranking-based strategies, we can observe the following. Independent of the total supply $f_v$ of some node $v$, there is a uniquely defined outgoing edge for the next unit of $v$'s supply. We call this edge the \emph{active edge} of $v$. If $\vecf \equiv \mathbf{0}$, the active edge of a node $v$ is $\pi'_v(1)$. If $\sum_{i=1}^{k-1}{c_{\pi'_v(i)}} \leq f_v < \sum_{i=1}^{k}{c_{\pi'_v(i)}}$ for some $k$, the active edge is $\pi'_v(k)$. 

In the following, we will describe the \TCI\ algorithm. For a formal description see Algorithm \ref{alg:tci}. In the description we distinguish necessary and optional cycles. This classification builds on the lattice structure of feasible flows (c.f.\ Theorem~\ref{thm:clearingLattice}). Since the algorithm computes the supremum of the lattice, it raises flow along all -- necessary and optional -- cycles.

\begin{algorithm}[t]
\DontPrintSemicolon

\SetKwInOut{Input}{Input}\SetKwInOut{Output}{Output}
\Input{The circulation network $G'$, strategy profile $\vecpi'$, auxiliary node $s$}
\Output{The circulation $\vecf$ representing the clearing state for $\vecpi$ in $G$}
$f_e \leftarrow 0$ for all $e \in E'\;.$\;
\While{there is a cycle $C$ of active edges under $\vecf$}{
 Choose an arbitrary cycle $C\;.$\;
 $\delta \leftarrow \min\limits_{e \in C}\{c_e - f_e\}$\;
 $f_e \leftarrow f_e + \delta$ \quad for all $e \in C\;.$\;
}
\Return{$\vecf$}
\caption{\TCI}\label{alg:tci}
\end{algorithm}

\paragraph{Necessary cycles.} 
Consider the set of all active edges for $\vecf \equiv \mathbf{0}$. Every node $v \neq s$ has exactly one outgoing active edge $\pi'_v(1)$. The set of active edges form disjoint cycles with attached trees. Each tree attached to a cycle is rooted in a node from the cycle, and directed towards the cycle $C$. We define the \emph{orbit} of $C$ by 
\[ o(C) = \{ v \in V \mid \exists \text{ $v$-$u$-path of active edges, for some } u\in C\},\]  
i.e., the set of nodes $v$ from which we can reach $C$ over active edges. Now, consider the auxiliary node $s$. Note that $\vecf \equiv 0$ is a feasible circulation in $G'$, but $\sum_v f_{(v,s)} = c^+_s = \sum_v b_v$ is a necessary constraint to ensure that $\vecf$ represents a feasible flow in the original graph $G$. However, as long as $\sum_v f_{(v,s)} < \sum_v b_v$ there is an active outgoing edge of $s$ and all nodes (including $s$) belong to some orbit. Hence, there is a cycle $C$ with $s \in o(C)$. Flow conservation and the monotonicity of strategies implies that some amount of flow of $s$ must eventually reach $C$. Due to flow conservation in $C$, a flow of at least $\delta_C = \min\{c_e \mid e \in C\}$ must thus be present on every edge of $C$. This is a necessary condition in every feasible flow $\vecf \in \mathcal{F}$.

It is straightforward to inductively apply this argument, thereby obtaining a sequence of necessary cycles $C$ that must be filled with flow $\delta_C = \min\{c_e - f_e \mid e \in C\}$. In particular, when a flow of $f_e = \delta_C$ has been assigned to every edge $e \in C$, the active edge of at least one of the vertices in $C$ changes. This implies that the orbits change, i.e., the orbit $o(C)$ partitions into new suborbits, or parts that get attached to other orbits. Note that once a vertex $v$ becomes fully saturated and all regular outgoing edges are filled, the active edge becomes $(v,s)$. 

Orbits present at the same time are always mutually disjoint. Thus, for two existing orbits $o(C_1)$ and $o(C_2)$, pushing flow along $C_1$ can never change the active edges of $o(C_2)$. Hence, it is necessary that all the cycles $C$, where $s$ eventually appears in the orbit, must get assigned a flow increase $\delta_C$ in order to reach $\sum_v f_{(v,s)} = \sum_v b_v$. Once we reach a flow $\vecf$ where $\sum_v f_{(v,s)} = \sum_v b_v$, an ``orbit'' $o(s)$ emerges composed of a tree rooted in $s$. At this point, we have indeed constructed a feasible flow, which by induction is the unique minimal circulation in $G'$ that represents a feasible flow in $G$.

\paragraph{Optional cycles.}
In the following, we characterize the structure of all other feasible flows $\vecf \in \mathcal{F}$ by applying similar observations. Fix some flow $\vecf$ in the circulation network that is also feasible in $G$, i.e., $\sum_v f_{(s,v)} = \sum_v b_v$ and consider the set of all active edges. 

Suppose there is a cycle $C$ with some orbit $o(C)$, i.e., a set of nodes $v$ that are not attached to the tree rooted in $s$. In this case, one can push a non-zero amount of flow along $C$, i.e., increase $f_v$ by a strictly positive amount for every $v \in C$. This obviously yields a new feasible clearing state. When an edge becomes saturated, the set of active edges changes and the orbit $o(C)$ disappears, i.e., gets split up as explained above (new suborbits, parts attached to other orbits, parts attached to the tree rooted in $s$).

Note that there might exist multiple orbits at the same time and, thus, multiple possibilities to extend $\vecf$ by increasing flow along a cycle. However, as observed above, orbits present at the same time are always mutually disjoint, and pushing flow along cycle $C_1$ can never change the active edges in an orbit $o(C_2)$ present at that time. Now consider some vertex $v \in o(C)$. In order to create some feasible flow $\vecf'$ with $f'_v > f_v$ it is necessary to push flow along $C$ until $f_v$ is reached (if $v \in o(C) \cap C$) or $o(C)$ disappears and splits up (if $v \in o(C) \setminus C$). Since the flow adjustments monotonically increase all flow values, there is a one-to-one correspondence between sets of cycles with flow increase and the feasible flows. 

Note that the cycles chosen for flow increase form a partial order: A cycle $C'$ might not be present in the beginning -- there might be predecessor-cycles $C$ that have to be filled up to $\delta_C$ to break an existing orbit $o(C)$, change some of the active edges, and make $C'$ appear. In the argumentation above, it can be seen that the set of predecessor cycles $pred(C')$ for some cycle $C'$ is uniquely defined resulting from the ranking of edges in $\vecpi$.

The arguments above imply a natural algorithm to compute $\vecfhat$, which is similar in spirit to the classic Top-Trading-Cycles algorithm for house allocation~\citep{ShapleyS74}. The \TCI\ algorithm iteratively raises flow along cycles among the active edges in the circulation network $G'$. Thereby it computes the unique maximal feasible flow $\vecfhat$ from the lattice. The algorithm runs in polynomial time -- in every round it increases the flow along a cycle $C$ by $\delta_C$. At this point at least one edge (from $G$ or auxiliary) becomes saturated. In terms of the original network $G=(V,E)$, the algorithm needs at most $O(|V| + |E|)$ rounds. Each round can easily be implemented in strongly polynomial time.

\begin{proposition}
\label{prop:clearingAlgo}
In every edge-ranking game, the \TCI\ algorithm computes the clearing state $\vecfhat$ in strongly polynomial time.
\end{proposition}

Since unit-ranking games can be cast as edge-ranking games with unit-capacity multi-edges, the algorithm can be applied in unit-ranking games as well. Note that the running time does not necessarily remain polynomial due to the pseudo-polynomial blowup in representation size.

\section{Unit-ranking games.}
\label{sec:coin}

\subsection{Existence and computation of equilibria.}

In this section we consider equilibria in unit-ranking games. Our first result is that in every unit-ranking game there is a strong equilibrium that can be computed in polynomial time. Moreover, for a large class of quality functions $\sw$, we can guarantee the existence of a strong equilibrium that is optimal with respect to $\sw$. In particular, every such equilibrium profile of unit-ranking strategies can be represented compactly. In the next theorem, we prove the existence and representation result, and we show that the price of stability is 1 for $\rev$. The extension to more general quality functions $\sw$ is discussed subsequently. 

Consider an arbitrary flow allocation game and a circulation that maximizes the total flow in the circulation network $G'$. We show that this circulation can be expressed as a clearing state of a strong equilibrium in unit-ranking strategies.

\begin{theorem}
\label{thm:coinSPoS1}
For every unit-ranking game, there is a strong equilibrium that maximizes the total amount of flow. The strong equilibrium can be computed in polynomial time, and the equilibrium profile can be represented in polynomial space.
\end{theorem}

\begin{proof}
Consider the circulation network $G'=(V,E')$. For a moment, assume this is a standard flow network without strategic flow allocation. Consider an optimal circulation $\vecf^*$ that maximizes the total flow value, i.e., it maximizes the sum of flow on all edges. This implies, in particular, that it saturates all outgoing auxiliary edges from $s$. Clearly, $\vecf^*$ yields an upper bound on the achievable total amount of flow (denoted by $\rev(\vecf^*)$) in any strategy profile $\veca$ of the game
\begin{align*}
\sum_{e \in E'} f_e^* &= 2\sum_{v \in V} b_v + \sum_{v \in V} \sum_{e \in E^+(v)} f^*_e\\
 &=  2\sum_{v \in V} b_v + \rev(\vecf^*) \geq 2\sum_{v \in V} b_v + \rev(\veca).
\end{align*}
$\vecf^*$ can be computed in strongly polynomial time~\citep{Tardos85}. Since all edge capacities are integral, we can assume all $f_e^*$ are integral. 

We now turn this circulation into a clearing state for a carefully chosen strategy profile $\veca^*$ of unit-ranking strategies. We will choose $\veca^*$ such that it can be compactly represented by \emph{threshold-ranking} strategies.

In a threshold-ranking strategy, every firm $v$ chooses a permutation $\pi_v$ over $E^+(v)$ and thresholds $\tau_e$. The interpretation of threshold-ranking strategies is that node $v$ first assigns $\tau_e$ particles to every edge $e \in E^+(v)$, sequentially in the order given by $\pi_v$. Then, it assigns the remaining $c_e - \tau_e$ particles to every edge in the order given by $\pi_v$. That is, $v$ first considers edge $\pi_v(1)$ and forwards the first $\tau_{\pi_v(1)}$ particles to this edge. The next $\tau_{\pi_v(2)}$ particles are forwarded to edge $\pi_v(2)$ etc.\ until $\sum_{j=1}^{|E^+(v)|} \tau_{\pi_v(j)}$ particles are sent to the edges (or $v$ runs out of flow). Then, the remaining $c_{\pi_v(1)} - \tau_{\pi_v(1)}$ particles are forwarded to edge $\pi_v(1)$, then the next $c_{\pi_v(2)} - \tau_{\pi_v(2)}$ particles to $\pi_v(2)$ etc. Clearly, threshold-ranking strategies are more general than edge-ranking strategies. They constitute a special class of unit-ranking strategies with compact representation.

We choose $\veca^*$ as follows. Every firm $v$ chooses an arbitrary permutation $\pi_v$ over $E^+(v)$ and sets $\tau_e = f^*_e$. It is easy to see that in $\veca^*$, the optimal circulation $\vecf^*$ corresponds to the clearing state. Let us prove that $\veca^*$ is a strong equilibrium, i.e., that no coalition $C \subseteq V$ has a profitable deviation.

Suppose for contradiction that there is a coalition $C$ with a profitable deviation. Examine the new profile $(\veca'_C, \veca^*_{-C})$ and assume $\vecfhat'$ is the clearing state. Consider a node $v \in C$. Since $u_v(\veca'_{C}, \veca^*_{-C}) > u_v(\veca^*)$, there must be strictly more outgoing flow from $v$ in the new profile. Due to the no-fraud condition, this can only happen if $v$ also has strictly more \emph{incoming flow} in the new profile. Hence, there is an incoming edge $e=(w,v) \in E^{-}(v)$ with $\hat{f}'_e > f^*_e$. Now consider node $w$. If $w \in C$, then $u_w(\veca'_{C}, \veca^*_{-C}) > u_w(\veca^*)$, so by the same reasoning there is again some incoming edge in $E^-(w)$ that has strictly more flow in $\vecfhat'$. Otherwise, if $w \not\in C$, then $w$ still plays the strategy $\veca_w^*$. Due to monotonicity, a higher flow on $(w,v)$ can only occur if $w$ has larger total supply. Thus, there is again some incoming edge in $E^-(w)$ that has strictly more flow in $\vecfhat'$.

We can repeat this argument indefinitely. As such, there must be a cycle of edges that all have more flow in $\vecfhat'$ than in $\vecf^*$. Such a cycle can be used to increase the circulation, which contradicts that $\vecf^*$ is an optimal circulation in $G'$.
\end{proof}

\begin{remark} \rm
For the profitable deviation, we can even allow arbitrary continuous strategies and any choice of clearing state for the deviation profile. This applies even for games with non-monotone and fraud strategies, as long as a feasible flow exists and the clearing state is chosen arbitrarily among the feasible flows that are not weakly dominated in terms of coordinate-wise comparison.
\end{remark}

\begin{remark} \rm
If we consider deviations that weakly improve the coalition (i.e., $u_v(\veca'_C,\veca_{-C}) \ge u_v(\veca)$ for all $v \in C$ and $u_w(\veca'_C,\veca_{-C}) > u_w(\veca)$ for at least one $w \in C$), it is a simple exercise to see that there are unit-ranking games, in which no such (often termed ``super-strong'') equilibrium exists.
\end{remark}

As mentioned before, the result can be generalized quite substantially beyond $\rev$ to a large class containing various quality functions $\sw(\veca)$. Let $\mathcal{A}$ be the set of all possible strategy profiles of a unit-ranking game. Let $\veca, \veca' \in \mathcal{A}$ be two strategy profiles and $\vecfhat$ and $\vecfhat'$ the clearing states in $\veca$ and $\veca'$, respectively. Recall the coordinate-wise comparison of flows used in Theorem~\ref{thm:clearingLattice}. The function $\sw : \mathcal{A} \to \RR$ is \emph{flow-monotone} if $\vecfhat > \vecfhat'$ implies $\sw(\veca) \ge \sw(\veca')$. 

\begin{corollary}
	\label{cor:coinSPoS1}
	For every unit-ranking game and every flow-monotone social welfare function $\sw$, there is a strong equilibrium that maximizes $\sw$. The equilibrium profile can be represented in polynomial space.
\end{corollary}

\begin{proof}
  The proof uses the same argument as above. Consider the optimal value $\sw_{\max} = \max_{\veca \in \mathcal{A}} \sw(\veca)$ and the set $\mathcal{A}_{\max} = \{ \veca \mid \sw(\veca) = \sw_{\max}\}$ of optimal strategy profiles w.r.t.\ $\sw$. Consider a profile $\veca^* \in \mathcal{A}_{\max}$ such that the clearing state $\vecfhat^*$ is not coordinate-wise dominated by a clearing state of another optimal profile, i.e., there is no $\veca' \in \mathcal{A}_{\max}$ with $\vecfhat' > \vecfhat^*$. We show that $\veca^*$ is a strong equilibrium.

  For contradiction, suppose there is a coalitional deviation in $\veca^*$. By the same arguments as in the proof above, the deviation implies that there is a cycle that allows to increase the flow. Thus, there is a circulation $\vecf > \vecfhat^*$ in the flow network $G'$. By constructing a strategy profile $\veca$ using threshold-ranking strategies, $\vecf$ can be turned into the clearing state of $\veca$. Since $\sw$ is flow-monotone, $\sw_{\max} = \sw(\veca^*) \le \sw(\veca)$. Hence, $\veca \in \mathcal{A}_{\max}$ and $\vecf > \vecfhat^*$, a contradiction to the choice of $\veca^*$. 
  
  Finally, by using threshold-ranking strategies the equilibrium profile has a compact representation.
\end{proof}

The corollary implies that the price of stability for strong equilibria is 1 for a very wide range of natural quality functions. For example, instead of $\rev(\veca)$ (representing utilitarian welfare~\eqref{eq:flowSW}) we might prefer to express the quality of a strategy profile by
\begin{compactitem}
\item \emph{egalitarian welfare} $\textsc{EW}(\veca) = \min_{v \in V} u_v(\veca)$, i.e., the minimum utility of any player in the network, or
\item \emph{Nash social welfare} $\textsc{NSW}(\veca) = \left(\prod_{v \in V} u_v(\veca)\right)^{1/n}$, i.e., the geometric mean of player utilities, or
\item the \emph{number of fully saturated nodes} $\textsc{FSN}(\veca) = |\{ v \in V \mid \hat{f}_v = c_v^+ \}|$, which in financial networks corresponds to the number of solvent firms, or
\item any monotone transformation or combination of the above functions.
\end{compactitem}

Corollary~\ref{cor:coinSPoS1} implies the existence of an optimal strong equilibrium for all these functions. 

While we can compute any strong equilibrium in polynomial time, for some functions $\sw$ an \emph{optimal} strong equilibrium can be \classNP-hard to compute. A simple reduction shows that for each of the objective functions in the list above it is \classNP-hard to compute an optimal strategy profile and, consequently, also an optimal strong equilibrium. 

\begin{theorem}
	\label{thm:coinPoSNPC}
	In unit-ranking games, it is strongly \classNP-hard to compute a strategy profile that maximizes $\textsc{EW}$, $\textsc{NSW}$, or $\textsc{FSN}$.
\end{theorem}

\begin{proof}
We start by considering the problem of optimizing \textsc{FSN}. This task can be at least as hard as deciding \textsc{Exact Cover by 3-Sets (X3C)}. In an instance of X3C, we have a set $R$ of $3k$ elements for an integer $k \ge 1$ and a set $\mathcal{S} \subseteq 2^{E}$ of $m$ triplets (i.e., $|S| = 3$ for each $S \in \mathcal{S}$). The goal is to decide if there are $k$ non-overlapping sets in $\mathcal{S}$. 

For the reduction, we construct a game containing a source node $v$ with $b_v = 3k$. Each $S \in \mathcal{S}$ is a node, and for each $r \in R$ there are three nodes $r^1,r^2,r^3$. We have edges $(v,S)$ of capacity $3k+1$, for every $S \in \mathcal{S}$, edges $(S,r^1)$ of capacity 1, for every $r \in S$, and two edges $(r^1,r^2), (r^2,r^3)$ of capacity 1 for every $r \in R$. Note that all $r^3$ are fully saturated, and $v$ can never be fully saturated. By routing the flow from $v$ via the sets to the elements, we can fully saturate at least the $9k$ nodes corresponding to elements in $R$. If nodes of overlapping sets are fully saturated, we route a flow of $x \ge 2$ to some element node $r^1$, which implies that $2(x-1)$ nodes for $x-1$ other elements $r' \in R$ are not fully saturated. This does not happen in an optimal profile if (and only if) the X3C instance is a yes-instance. More formally, an optimal profile saturates at least $9k+k$ nodes for elements and sets if and only if there are $k$ non-overlapping sets in $\mathcal{S}$.

To prove the result for objectives \textsc{EW} and \textsc{NSW}, we introduce two auxiliary nodes $S_a$ and $r_a$. For each $S \in \mathcal{S}$, we add edges $(S,S_a)$ and $(S_a,S)$ with capacity 1. W.l.o.g.\ we can assume that all nodes $S$ give first priority to the edge $(S,S_a)$. Consequently, each $S$ has a positive utility of at least 1. $S_a$ has utility $|\mathcal{S}|$. Similarly, for each node $r^3$, we add edges $(r^3,r_a)$ and $(r_a,r^3)$ with capacity 1. As such, all nodes $r^3$ can be assumed to have a utility of at least 1. $r_a$ has utility $3k$. 

Now the problem of optimizing \textsc{EW} or \textsc{NSW} reduces to deciding whether all nodes $r^1, r^2$ can simultaneously obtain positive utility, since otherwise the objective value will be 0. This property, however, is equivalent to choosing $k$ non-overlapping sets to route the total demand of $3k$ from $v$ to the $3k$ nodes $r^1$. As such, a positive value for \textsc{EW} or \textsc{NSW} can be obtained if and only if the X3C instance is a yes-instance.
\end{proof}

Irrespective of the quality function, even computing a \emph{best-response strategy} for a single node $v$ in a strategy profile can be strongly \classNP-hard, since best responses can provide answers to computationally hard decision problems. This holds even in games without fixed supply and with edge capacities in $\{0,1\}$.

\begin{theorem}
	For a given strategy profile $\veca$ of a unit-ranking game with $b_v = 0$ for all $v \in V$ and $c_e \in \{0,1\}$ for all $e \in E$, deciding whether a given node $v$ has a best response resulting in utility at least $k$ is strongly \classNP-complete.
	\label{thm:bestRespNPC}
\end{theorem}
\begin{figure}[ht!]
	\begin{center}
		\begin{subfigure}[t]{0.4\linewidth}
			\begin{tikzpicture}[auto,swap,scale=1.1]
				\node[vertex](s) at (0,0){$v$};
				\node[vertex](xi10) at (-1.879,-.684){$x_{i,1,0}$};
				
				\node[smallvertex](zi0) at (-.9,-.7){$z_{i,0}$};
				\node[smallvertex](zi1) at (-.9,.7){$z_{i,1}$};
				
				\node[vertex](xi11) at (-1.879,.684){$x_{i,1,1}$};
				\node[vertex](xi20) at (-1,-1.732){$x_{i,2,0}$};
				\node[vertex](xi21) at (-1,1.732){$x_{i,2,1}$};
				\node[vertex](xi30) at (0.347,-1.97){$x_{i,3,0}$};
				\node[vertex](xi31) at (0.347,1.97){$x_{i,3,1}$};
				\node[vertex](xi40) at (1.532,-1.285){$x_{i,4,0}$};
				\node[vertex](xi41) at (1.532,1.285){$x_{i,4,1}$};
				\node[vertex](xi) at (2,0){$z_{i}$};
				\path[edge] (s) edge [bend right = 0] (xi10);
				\path[edge] (s) edge (xi20);
				\path[edge] (s) edge (xi21);
				\path[edge] (s) edge (xi30);
				\path[edge] (s) edge (xi31);
				\path[edge] (s) edge (xi40);
				\path[edge] (s) edge (xi41);
				%
				
				\path[edge] (s) edge (zi0);
				\path[edge] (zi0) edge (xi10);
				
				\path[edge] (xi10) edge (xi20);
				\path[edge] (xi20) edge (xi30);
				\path[edge] (xi30) edge (xi40);
				\path[edge] (xi40) edge (xi);
				\path[edge] (s) edge [bend right = 0] (xi11);
				
				\path[edge] (s) edge (zi1);
				\path[edge] (zi1) edge (xi11);
				
				\path[edge] (xi11) edge (xi21);
				\path[edge] (xi21) edge (xi31);
				\path[edge] (xi31) edge (xi41);
				\path[edge] (xi41) edge (xi);
				\path[edge] (xi) edge (s);
			\end{tikzpicture}
			\caption{Variable Gadget}
			\label{fig:varGadget}
		\end{subfigure}%
		\begin{subfigure}[t]{0.5\linewidth}
			\begin{center}
				\begin{tikzpicture}[auto,swap,scale=0.85]
					\node[vertex](s) at (0,0){$v$};
					\node[vertex](c1) at (.75,1.5){$c_1$};
					\node[vertex](c2) at (1.5,-.75){$c_2$};
					\node[vertex](c3) at (-.75,-1.5){$c_3$};
					\node[vertex](c4) at (-1.5,.75){$c_4$};
					\node[vertex](x111) at (1,3.5){$x_{1,1,1}$};
					\node[vertex](x211) at (2,3){$x_{2,1,1}$};
					\node[vertex](x310) at (2.5,2){$x_{3,1,0}$};
					\node[vertex](x121) at (3,0){$x_{1,2,1}$};
					\node[vertex](x220) at (3.25,-1.5){$x_{2,2,0}$};
					\node[vertex](x421) at (2,-2.5){$x_{4,2,1}$};
					\node[vertex](x331) at (-1,-3.5){$x_{3,3,1}$};
					\node[vertex](x430) at (-2,-3){$x_{4,3,0}$};
					\node[vertex](x241) at (-2.5,0){$x_{2,4,1}$};
					\node[vertex](x340) at (-3,1){$x_{3,4,0}$};
					\node[vertex](x441) at (-2.5,2){$x_{4,4,1}$};
					\node[vertex](x541) at (-1.5,2.5){$x_{5,4,1}$};
					\path[edge] (c1) edge (s);
					\path[edge] (c2) edge (s);
					\path[edge] (c3) edge (s);
					\path[edge] (c4) edge (s);
					\path[edge] (x111) edge (c1);
					\path[edge] (x211) edge (c1);
					\path[edge] (x310) edge (c1);
					\path[edge] (x121) edge (c2);
					\path[edge] (x220) edge (c2);
					\path[edge] (x421) edge (c2);
					\path[edge] (x331) edge (c3);
					\path[edge] (x430) edge (c3);
					\path[edge] (x241) edge (c4);
					\path[edge] (x340) edge (c4);
					\path[edge] (x441) edge (c4);
					\path[edge] (x541) edge (c4);
				\end{tikzpicture}
			\end{center}
			\caption{Example Construction of a Network}
			\label{fig:restGraphNPHardness}
		\end{subfigure}
	\end{center}
	\caption{Structures used in the proof of Theorem~\ref{thm:bestRespNPC}.}
\end{figure}%

\begin{proof}
For unit-ranking games with edge capacities in $\{0,1\}$ and fixed supplies equal to $0$, the decision problem is obviously contained in \classNP: We can represent every unit-ranking strategy as a ranking over edges. Then, our \TCI\ algorithm to compute $\vecfhat$ discussed in Section~\ref{sec:topCycle} runs in polynomial time and we can efficiently verify the utilities for all players.

For strong \classNP-hardness, suppose we are given an instance $I$ of \SAT\ in conjunctive normal form with $n$ variables and $m$ clauses. We construct a unit-ranking game in edge-ranking representation with a node $v$ and a strategy profile $\vecpi_{-v}$ for the other players such that the following holds: There is a strategy $\pi_v$ with utility $u_v(\pi_v, \vecpi_{-v}) \ge k'+n$ if and only if $I$ has a variable assignment that fulfills at least $k'$ clauses.

We construct the game as follows. We denote the variables of $I$ by $x_1,\ldots,x_n$ and the clauses by $C_1,\ldots,C_m$. For each variable $x_i$ we create nodes $x_{i,j,0}$ and $x_{i,j,1}$ for all $j \in \{1,\dots,m\}$, as well as a node $z_i$. For each clause $C_j$ we add a clause node $c_j$. In addition, there is a separate node $v$, for which we strive to find a best response.

For each clause $C_j$, we add a unit-capacity edge from $x_{i,j,0}$ to $c_j$ if $x_i$ appears as $\neg x_i$ in $C_j$ and from $x_{i,j,1}$ to $c_j$ if it appears as $x_i$ in $C_j$. Flow incoming to $c_j$ will eventually indicate a literal that fulfills the clause $C_j$. There is an edge $(c_j,v)$ for all $j \in \{1,\dots, m\}$. We will show below that this edge ensures that satisfying clause $C_j$ adds exactly one unit to the total supply of $v$.

For each variable $x_i$, we add a variable gadget. It consists of nodes $v, x_{i,j,0}$ and $x_{i,j,1}$ for all $j\in \{1,\dots,m\}$, as well as auxiliary nodes $z_i$, $z_{i,0}$ and $z_{i,1}$. There are unit-capacity edges $(v,x_{i,j,0})$ and $(v,x_{i,j,1})$ for all $j \in \{1,\dots,m\}$, edges $(x_{i,j,0},x_{i,j+1,0})$, $(x_{i,j,1},x_{i,j+1,1})$ for $j \in \{1, \dots, m-1\}$, and edges $(v,z_{i,0})$, $(z_{i,0}, x_{i,1,0})$ and $(v,z_{i,1})$, $(z_{i,1},x_{i,1,1})$. Firm $z_i$ has edges $(x_{i,m,0},z_i)$, $(x_{i,m,1},z_i)$, and $(z_i,v)$. An example for the gadget that is constructed for a variable $x_i$ and $m=4$ is depicted in Fig.~\ref{fig:varGadget}. Note that for every strategy of $v$, there is at most one cycle emerging in this gadget, since all cycles must include the outgoing edge of $z_i$. 

In Fig.~\ref{fig:restGraphNPHardness}, we show an example of the network without the variable gadgets for the \SAT\ instance $I = (x_1 \vee x_2 \vee \neg x_3) \wedge (x_1 \vee \neg x_2 \vee x_4) \wedge (x_3 \neg x_4) \wedge (x_2 \vee \neg x_3 \vee x_4 \vee x_5)$.

We construct a strategy profile $\vecpi_{-v}$ as follows. Observe that nodes $c_1,\ldots,c_m$ and $x_1,\ldots,x_n$ each have a single outgoing edge, their strategies are trivial. If a node $x_{i,j,0}$ or $x_{i,j,1}$ has multiple outgoing edges, it always prioritizes the edge to nodes $x_{i,j+1,0}$ and $x_{i,j+1,1}$, respectively, or to $z_i$ if $j=m$. 

In the following, we argue that there is a best response of $v$ with utility $k'+n$ if and only if there is a variable assignment such that $k'$ clauses are fulfilled in $I$. Suppose there is a variable assignment such that $k'$ clauses are fulfilled. Fix this assignment, and for every satisfied clause $c$ choose a single literal $l_c$ that evaluates to true in the clause. Let $v$ choose the following strategy: First, prioritize edges $(v,x_{i,1,0})$ if $x_a = \text{false}$ in the assignment and $(v,x_{i,1,1})$ if $x_a = \text{true}$ in the assignment. All these edges will close a cycle $(v, x_{i,1,0}, x_{i,2,0}, \dots, x_{i,m,0}, x_i, v)$ or $(v, x_{i,1,1}, x_{i,2,1}, \dots, x_{i,m,1}, x_i, v)$. After that, for all clause-fulfilling literals $l_c$ prioritize the edges $(v,x_{i,j,0})$ if $c=C_j$ and $l_c = \neg x_i$ in any order. Prioritize the edges $(v,x_{i,j,1})$ if $c=C_j$ and $l_c = x_i$. All these edges close a cycle via the clause node $c_j$, leading to a total inflow of $k'+n$.

For showing the other direction, we observe the following structural property for all variables $x_i$. If there is some flow on an edge $(x_{i,m,0},z_i)$ there cannot be any flow on edge $(x_{i,m,1},z_i)$, and vice versa. We conclude that flow on some edge $(x_{i,j,0},c_j)$ implies flow on edge $(x_{i,j,0},x_{i,j+1,0})$ (since it has a higher priority), and $(x_{i,m,0},z_i)$, and thus no flow on $(x_{i,j',1},c_{j'})$ for all $j'$. Analogously, we observe that any flow on some edge $(x_{i,j,0},c_j)$ implies no flow on edges $(x_{i,j',1},c_{j'})$.

We observe that if there is a best response of node $v$ with inflow equal to $k'+n$, the node has $k'+n$ incoming edges that carry flow. At most $n$ of these edges can be $(z_i,v)$-edges, so there are at least $k'$ clause-edges $(c_j,v)$ that carry flow. Thus, all these $k'$ clause nodes receive incoming flow. If this flow for some clause $c_j$ comes from a node $x_{i,j,0}$, we know by the observation above that no edge $x_{i,j',1}$ carries flow. Thus we can set the variable $x_i$ to false, which fulfills clause $C_j$. Applying this and analogous operations for flow on an edge $(x_{i,j,1},v)$ yields a variable assignment which fulfills at least $k'$ clauses.
\end{proof}

\subsection{Total flow in equilibrium.}

In this section, we analyze the quality of pure Nash and strong equilibria in unit-ranking games, mostly in terms of the \rev\ objective. In the last section we observed that the prices of stability for Nash and strong equilibria in unit-ranking games are both 1. We here bound the prices of anarchy for Nash and strong equilibria. The total flow depends crucially on the emergence of cycles in the strategy profile. This requires an effort that is inherently coalitional. As such, it might be unsurprising that there are games in which the worst Nash equilibrium may fail to provide any reasonable fraction of the optimal total flow.

\begin{proposition}
\label{prop:poaUnbounded}
The price of anarchy for Nash equilibria in terms of \rev\ is unbounded, even in unit-ranking games without fixed supplies.
\end{proposition}

\begin{proof}
Consider the game depicted in Fig.~\ref{fig:poaunbounded}. All edges have capacity 1, fixed supplies are $b_v = 0$ for all nodes $v$. Consider $\vecpi$ with $\pi_1 = (e_1, e_3)$ and $\pi_2=(e_2,e_4)$. It is a pure Nash equilibrium with $\rev(\vecpi) = 0$. No unilateral deviation can close a cycle and increase the value of $\vecfhat$. The optimal solution $\vecpi^*$ with $\pi^*_1 = (e_3, e_1)$ and $\pi^*_2=(e_4,e_2)$ has $\rev(\veca^*) = 2$.
\end{proof}
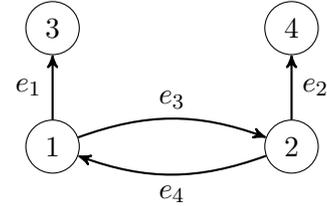
\begin{figure}[ht]
\begin{center}
\begin{scaletikzpicturetowidth}{0.45\textwidth}
\begin{tikzpicture}[scale=\tikzscale,auto,swap]
\node[vertex](v1) at (0,0){$1$};
\node[vertex](v2) at (3,0){$2$};
\node[vertex](v3) at (0,1.5){$3$};
\node[vertex](v4) at (3,1.5){$4$};
\path[edge] (v1) edge [bend left = 20,above] node{$e_3$} (v2);
\path[edge] (v2) edge [bend left = 20,below] node{$e_4$} (v1);
\path[edge] (v1) edge [left] node {$e_1$}(v3);
\path[edge] (v2) edge node{$e_2$}(v4);
\end{tikzpicture}
\end{scaletikzpicturetowidth}
\caption{The graph used in the proof of Porposition~\ref{prop:poaUnbounded}.}
\label{fig:poaunbounded}
\end{center}
\end{figure}
\null

To analyze the quality of strong equilibria, we again consider the unit-ranking game in the form of unit-capacity multi-edges. Consider an optimal circulation $\vecf^*$ of maximum social welfare in the circulation network $G'$. Since we have unit-capacity edges, we can assume that the optimal circulation has binary flows on each edge. Let $\mathcal{C}(\vecf^*) = \{C_1,\ldots,C_k\}$ be a decomposition of $\vecf^*$ into cycles of unit flow. We denote by 
\[ d = \min_{\vecf^*, \mathcal{C}(\vecf^*)} \; \max_{C \in \mathcal{C}(\vecf^*)} |C_i| \]
the min-max size of any cycle, in any decomposition $\mathcal{C}(\vecf^*)$ of any optimal circulation $\vecf^*$.

\begin{theorem}
\label{thm:spoaD}
In unit-ranking games, the price of anarchy for strong equilibria in terms of \rev\ is at most $d$.
\end{theorem}

\begin{proof}
Consider an optimal circulation $\vecf^*$ and a decomposition $\mathcal{C}(\vecf^*)$ such that all flow cycles $C_i \in \mathcal{C}(\vecf^*)$ have size at most $|C_i| \le d$. As observed in the proof of Theorem~\ref{thm:coinSPoS1}, this circulation yields the total flow of an optimal strategy profile $\veca^*$, i.e.,  
\[
\rev(\veca^*) = \sum_{v \in V} \sum_{e \in E^+(v)} f^*_e = \sum_{C_i \in \mathcal{C}(\vecf^*)} |C_i| - 2\sum_{v \in V} b_v \le \sum_{C_i \in \mathcal{C}(\vecf^*)} d \enspace.
\]

Now consider any strong equilibrium $\veca$ in the unit-ranking game with clearing state $\vecfhat$. The flow $\vecfhat$ can be assumed to have binary edge flows. Suppose there is a cycle $C_i \in \mathcal{C}(\vecf^*)$ such that $a_e(\vecfhat_v) = 0$ for all $e = (v,w) \in C_i$. Then the nodes in this cycle have an incentive to jointly deviate and place the edges of $C_i$ on the first position in their ranking. Then the clearing state $\vecfhat$ will emerge as before, adding a flow of 1 along the cycle $C_i$. This is a profitable deviation for the nodes of $C_i$. 

Consequently, for every cycle $C_i \in \mathcal{C}(\vecf^*)$ there must be at least one edge $e = (u,v) \in C_i$ such that $a_e(\vecfhat_u) = 1$. Thus, the total flow in the strong equilibrium $\veca$ is 
\[
\rev(\veca) \ge \sum_{C_i \in \mathcal{C}(\vecf^*)} 1
\]
and, hence, the ratio is at most $d$.
\end{proof}

\begin{figure}[t]
\centering
\begin{tikzpicture}[
  point/.style={circle,draw,minimum size=#1},
  point/.default=0pt 
]
  \foreach[count=\i,evaluate=\i as \angle using (\i-1)*360/5-54] \text in {%
      $v_1$,$v_2$,$v_3$,$v_{4}$,$v_5$
    }
    \node[vertex] (\i) at (\angle:1.5) {\small \text};
  \path[edge] (1) edge (2);
  \path[edge] (2) edge (3);
  \path[edge] (3) edge (4); 
  \path[edge] (4) edge (5);
  \path[edge] (5) edge (1);
  
    \node[vertex] (22) at ($(2.2,-0.8) + (-18:1.5)$) {\small $v_{2}^{2}$};
    \node[vertex] (21) at ($(2.2,-0.8) + (54:1.5)$) {\small $v_{2}^{1}$};
    \node[vertex] (23) at ($(2.2,-0.8) + (270:1.5)$) {\small $v_{2}^{3}$};

  \path[edge] (2) edge (21);
  \path[edge] (21) edge (22);
  \path[edge] (22) edge (23); 
  \path[edge] (23) edge (1);
      
    \node[vertex] (32) at ($(1.4,2) + (54:1.5)$) {\small $v_{3}^{2}$};
    \node[vertex] (31) at ($(1.4,2) + (126:1.5)$) {\small $v_{3}^{1}$};
    \node[vertex] (33) at ($(1.4,2) + (-18:1.5)$) {\small $v_{3}^{3}$};
    
  \path[edge] (3) edge (31);
  \path[edge] (31) edge (32);
  \path[edge] (32) edge (33); 
  \path[edge] (33) edge (2);

    \node[vertex] (42) at ($(-2.2,-0.8) + (198:1.5)$) {\small $v_{5}^{2}$};
    \node[vertex] (41) at ($(-2.2,-0.8) + (270:1.5)$) {\small $v_{5}^{1}$};
    \node[vertex] (43) at ($(-2.2,-0.8) + (126:1.5)$) {\small $v_{5}^{3}$};
    
  \path[edge] (5) edge (41);
  \path[edge] (41) edge (42);
  \path[edge] (42) edge (43); 
  \path[edge] (43) edge (4);

    \node[vertex] (52) at ($(-1.4,2) + (126:1.5)$) {\small $v_{4}^{2}$};
    \node[vertex] (51) at ($(-1.4,2) + (198:1.5)$) {\small $v_{4}^{1}$};
    \node[vertex] (53) at ($(-1.4,2) + (54:1.5)$) {\small $v_{4}^{3}$};
    
  \path[edge] (4) edge (51);
  \path[edge] (51) edge (52);
  \path[edge] (52) edge (53); 
  \path[edge] (53) edge (3);

\end{tikzpicture}
\caption{A unit-ranking game with $d=5$ and a price of anarchy for strong equilibria of $d-1 = 4$.}
\label{fig:strongPoA}
\end{figure}

\begin{proposition}
\label{prop:spoa_unit_edge}
For every $d \ge 2$, there is a unit-ranking game in which the price of anarchy for strong equilibria in terms of \rev\ is at least $d-1$.
\end{proposition}
\begin{proof}
The game is given by a graph $G$ with $d+(d-1)(d-2)$ nodes. $G$ is constructed as follows. The nodes $v_1,\dots,v_d$ are called \emph{central nodes} and they form a cycle of length $d$. For each $i=1,\dots,d-1$, there are nodes $(v_{i,j})_{j = 1,\dots,d-2}$ that form additional cycles of length $d$ with the edge $(v_i,v_{i+1})$. Thus, the set of edges is given by
\begin{align*}
E&=\{(v_i,v_{i+1}) \mid i \in \{1,\dots,d-1\}\} \cup \{(v_d,v_1)\}\\
& \quad \cup \bigcup_{i=2,\dots,d} \big((v_{i},v_{i}^{1}) \cup \{(v_{i}^{j},v_{i}^{j+1}) \mid j=1,\dots,d-3\} \cup (v_{i}^{d+2},v_{i-1})\big)\;.
\end{align*}
All nodes have fixed supply of 0. All edges have unit capacity. An example of the instance with $d = 5$ is depicted in Fig.~\ref{fig:strongPoA}. Observe that only nodes $v_i$, $i=2,\dots,d$ have multiple outgoing edges and, thus, these nodes are the only ones with a non-trivial strategy choice.

Since all edges have capacity 1, we will view this game equivalently as an edge-ranking game. We will show that the strategy profile $\pi_i=\big((v_i,v_{i+1}),(v_i,v_{i}^{1})\big)$, for $i=2,\ldots,d-1$, and $\pi_d=\big((v_d,v_{1}),(v_d,v_{d}^{1})\big)$ is a strong equilibrium in the game. In order to see this, let $\vecfhat$ be the clearing state of $\vecpi$. Note that
\[\vecfhat_v = \begin{cases}1 & \text{ if }v=v_1,v_2,\dots,v_d\;,\\ 0 & \text{ otherwise\;,}\end{cases}\]
that is, $\rev(\vecpi) = d$. Now suppose there is a non-empty coalition of nodes $C \subseteq (v_2,\dots,v_d)$ that all strictly increase their utility by a joint deviation. Note that $v_d$ only has a single incoming edge that is saturated in $\vecpi$, so $v_d \notin C$. Thus, the cycle $v_d, v_{d}^{1}, \dots, v_{d}^{d-2},v_{d-1}$ cannot carry any flow. We conclude that $v_{d-1}$ has only a single edge that can carry flow. Iterating this argument yields $C = \emptyset$, a contradiction.

In constrast, the strategy profile with $\pi_i=\big((v_i,v_{i}^{1}),(v_i,v_{i+1})\big)$ for $i=2,\ldots,d-1$ and $\pi_d=\big((v_d,v_{d}^{1}),(v_d,v_{1})\big)$ has social welfare of $(d-1)d$. Thus, the price of anarchy for strong equilibria in this instance is at least $d-1$.
\end{proof}

For completeness, let us also provide a cumulative characterization of prices of anarchy for pure Nash and strong equilibria for the three other quality functions mentioned in the previous section.

\begin{proposition}
	\label{prop:otherPoA}
	In unit-ranking games, the prices of anarchy for pure Nash and strong equilibria in terms of 
	\begin{itemize}
		\item \textsc{EW} are unbounded,
		\item \textsc{NSW} are unbounded,
		\item \textsc{FSN} are exactly $n-1$.
	\end{itemize}
	The lower bounds apply even in games with $d=3$.
\end{proposition}

\begin{proof}
Consider the game depicted in Fig.~\ref{fig:triangle}. All edges have capacity 1, fixed supplies are $b_v = 0$ for all nodes. Only $v_1$ has a non-trivial strategy choice, $\pi_1 = (e_1, e_2)$ or $\pi_2=(e_2,e_1)$. Both choices yields a utility of 1 for $v_1$, so both are a best response for $v_1$. Since $v_1$ is the only player with a choice, both options represent a strong equilibrium.
\begin{figure}[ht]
	\begin{center}
		\begin{scaletikzpicturetowidth}{0.35\textwidth}
			\begin{tikzpicture}[scale=\tikzscale,auto,swap]
				\node[vertex](v1) at (0,0){$v_1$};
				\node[vertex](v2) at (3,0){$v_2$};
				\node[vertex](v3) at (1.5,1.5){$v_3$};
				\path[edge] (v1) edge [bend left = 20,above] node{$e_1$} (v2);
				\path[edge] (v2) edge [bend left = 20,below] node{$e_4$} (v1);
				\path[edge] (v1) edge [left] node {$e_2$}(v3);
				\path[edge] (v3) edge [right] node{$e_3$}(v2);
			\end{tikzpicture}
		\end{scaletikzpicturetowidth}
	\end{center}
	\caption{The graph used in the proof of Proposition~\ref{prop:otherPoA}.}
	\label{fig:triangle}
\end{figure}

Both \textsc{EW} and \textsc{NSW} are 0 in the strong equilibrium when $v_1$ plays $\pi_1$ and 1 in the optimal state with $\pi_2$. Hence, the prices of anarchy for both pure Nash and strong equilibria are unbounded for both these objectives. In the optimal state, the unique flow cycle has length 3. Towards the \textsc{FSN} objective, we note that in every strategy profile $\veca$ of every unit-ranking game there must be at least one node that is fully saturated. Otherwise, every node has an outgoing active edge, and hence there must be cycle of active edges -- a contradiction to the property that the clearing state is maximal. Clearly, since there can be at most $n$ saturated nodes, the prices of anarchy for pure Nash and strong equilibria are at most $n$. 

More precisely, suppose there is a state $\veca^*$ with \textsc{FSN}$(\veca^*) = n$. In Proposition~\ref{prop:solventUnique} we proved that the strategy choices of saturated nodes have no impact on the clearing state. As such, if \emph{all} nodes are fully saturated, then for each node her utility is independent of the entire strategy profile. In this case, prices of anarchy for pure Nash and strong equilibria are 1. The prices are non-trivial only if every strategy profile has at least one non-saturated node. Hence, they can be at most $n-1$.

To show the lower bound of $n-1$, we slightly extend the construction above. We add nodes $v_4,\ldots,v_n$, each with two edges $(v_1,v_i)$ and $(v_i,v_2)$ of capacity 1. Edges $(v_1,v_2)$ and $(v_2,v_1)$ get an increased capacity of $n-2$. Again, $v_1$ is the only player with a non-trivial strategy choice. It is easy to see that for any strategy, $v_1$ obtains a utility of $n-2$. As such, every strategy is a best response, and every state represents a pure Nash and a strong equilibrium. In the worst equilibrium $\veca$, $v_1$ exchanges $n-2$ units of flow with $v_2$. $v_2$ gets saturated and \textsc{FSN}$(\veca) = 1$. In the optimal state $\veca^*$, $v_1$ routes $n-2$ units of flow in single units via $v_3,\ldots,v_n$ to $v_2$. Then $v_2,\ldots,v_n$ get saturated and \text{FSN}$(\veca^*) = n-1$. Note that in $\veca^*$ every cycle has length equal to 3.
\end{proof}

\section{Edge-ranking games.}
\label{sec:edge}

\subsection{Existence and computation of equilibria.}

With unit-ranking strategies we assume that nodes have flexibility in allocation of single particles. In this section, we focus on strategies, in which nodes simply rank their outgoing edges and allocate flow in order of this ranking until they run out of supply or all edges are saturated. In contrast to unit-ranking games, the restriction to rankings over edges (with different capacity) can destroy the existence of (optimal) stable states. In fact, there are even games without a single pure Nash equilibrium.
\begin{figure}[ht]
\begin{center}
\begin{tikzpicture}[auto,swap,scale=0.99]
\node[vertex](v1) at (0,0){$v_1$};
\node[vertex](v4) at (-2,-2){$v_4$};
\node[vertex,label=above:\fbox{2}](v2) at (-4,0){$v_2$};
\node[vertex](v6) at (-6,0){$v_5$};
\node[vertex](v7) at (2,-2){$v_6$};
\node[vertex,label=above:\fbox{2}](v3) at (4,0){$v_3$};
\node[vertex](v9) at (6,0){$v_7$};
\path[edge] (v1) edge [left] node{$4$} (v4);
\path[edge] (v4) edge [left] node {$2$}(v2);
\path[edge] (v2) edge [above] node{$6$}(v1);
\path[edge] (v2) edge [bend left = 20, below] node{$6$} (v6);
\path[edge] (v6) edge [bend left = 20, above] node{$1$} (v2);
\path[edge] (v1) edge node[right] {$4$} (v7);
\path[edge] (v7) edge [right] node {$2$}(v3);
\path[edge] (v3) edge [above] node{$6$}(v1);
\path[edge] (v3) edge [bend right = 20, below] node{$6$} (v9);
\path[edge] (v9) edge [bend right = 20, above] node{$1$} (v3);
\end{tikzpicture}
\end{center}
\caption{An edge-ranking game without a pure Nash equilibrium.}
\label{fig:noNash}
\end{figure}

\begin{proposition}
There is an edge-ranking game without a pure Nash equilibrium.
\label{prop:noNash}
\end{proposition}
\begin{proof}
Consider the game in Fig.~\ref{fig:noNash}. The capacities of the edges are depicted next to the edges. Nodes $v_2$ and $v_3$ each have fixed supply of $2$, the other nodes have fixed supply 0. Nodes $v_1$, $v_2$ and $v_3$ are the only ones with multiple outgoing edges and, thus, a non-trivial strategy choice. Due to the symmetry of the graph, we can assume w.l.o.g.\ $\pi_{v_1}=((v_1,v_4),(v_1,v_6))$. There are two possible strategy choices for each of the nodes $v_2$ and $v_3$. Checking all four resulting strategy profiles yields the following utility matrix for nodes $v_2$ and $v_3$:

\begin{center}
\begin{tabular}{l || lcr | lcr}
\diagbox[innerwidth=3cm,innerleftsep=.7cm,innerrightsep=.7cm]{$\pi_{v_2}$}{$\pi\textsl{}_{v_3}$} & \multicolumn{3}{c|}{$((v_3,v_1),(v_3,v_7))$} & \multicolumn{3}{c}{$((v_3,v_7),(v_3,v_1))$} \\\hline\hline
&   & & 4 &   & & 3\\
$((v_2,v_1),(v_2,v_5))$ & & & & & & \\
& 4 & &   & 4 & &  \\\hline
&   & & 2 &   & & 3\\
$((v_2,v_5),(v_2,v_1))$  & & & & & & \\ 
& 5 & &   & 3 & &  \\
\end{tabular}
\end{center}

\medskip

Inspecting the utilities, we see that there is no pure Nash equilibrium.
\end{proof}

The next theorem shows that a number of natural decision and optimization problems in edge-ranking games are indeed computationally intractable. Note that for unit-ranking games, these problems are either trivial (a strong equilibrium always exists) or can be solved in polynomial time (compute a strong equilibrium that represents a profile with maximum total flow).

\begin{theorem}
\label{thm:edgeExistNPC}
In an edge-ranking game the following problems are strongly \classNP-hard:
\begin{compactenum}
\item Deciding whether a pure Nash equilibrium exists or not.
\item Deciding whether a strong equilibrium exists or not.
\item Computing a pure Nash equilibrium, when it is guaranteed to exist.
\item Computing a strong equilibrium, when it is guaranteed to exist.
\item Computing a strategy profile $\vecpi$ that maximizes $\rev(\vecpi)$.
\item For a given strategy profile $\vecpi$, deciding whether a given node has a best response resulting in utility at least $k$.
\end{compactenum}
\end{theorem}

\begin{proof}
We start by proving hardness for computing a pure Nash or a strong equilibrium when it is guaranteed to exist. 

\paragraph{Hardness of computing Nash or strong equilibria.}
Consider an instance $I$ of the problem \ThreeDM. $I$ is given by a finite set $T$ with $|T|=3k$ and a set $U \subseteq T \times T \times T$. \cite{Karp72} proved it is strongly \classNP-complete to decide whether there is a subset $W \subseteq U$ such that $|W| = k$ and no two elements of $W$ have a non-empty intersection. The existence of such a set $|W|$ would be an exact cover of $T$. Given an instance $I$ of \ThreeDM, we construct an edge-ranking game as follows. Suppose there is a central node $v$. For each set $u \in U$, add a node $u$ and connect $v$ to $u$ by an edge $(v,u)$ with capacity $c((v,u))=3$. For each pair $t,u$ with $t \in T$, $u \in U$ and $t \in u$ add a node $t$ and add an edge $(u,t)$ with capacity $c((u,t))=1$. Finally, connect each element node $t \in T$ by an edge $(t,v)$ with $c((t,v))=1$ to $v$. 

The idea is that computing any pure Nash or strong equilibrium reduces to finding a best response for player $v$. This best response, however, gives utility of $3k$ if and only if $I$ has a solution. Hence, by computing a pure Nash or strong equilbrium, we obtain a best response for $v$ and thereby a certificate as to whether $I$ is solvable or not (and vice versa).

Let us first argue that there is always a strong equilibrium in this edge-ranking game. First, we note that all nodes $t \in T$ have a single outgoing edge. We fix an arbitrary feasible strategy vector $\vecpi'$. We will argue that best-response dynamics yield a strong equilibrium. For every strategy of $v$, there is at most one node $u_i$ with $a_{u_i} \in \{1,2\}$. For all other nodes $u \in U \setminus \{u_i\}$, every strategy $pi_u$ is a best response for $u$ due to Proposition~\ref{prop:solventUnique}. We conclude that if $\vecpi'$ is not a strong equilibrium, there is a coalition $S$ of players with an improvement move such that all players in $S$ strictly increase their utility. If $v \notin S$, there is a player $u_i \in S \cup U$. The improving move of $u_i$ also increases the utility of $v$. Thus, the utility of $v$ increases in every step of the dynamics. This shows that the dynamics terminate with a strong equilibrium.

Let $\vecpi$ be any Nash equilibrium in the edge-ranking game. We show the following claim. There is a subset $W$ such that $|W| = k$ and no two elements of $W$ have a non-empty intersection in $I$ if and only if $u_v(\vecpi) = |T|$ in $\vecpi$.

First, let us assume that $\vecpi$ is a Nash equilibrium with $u_v(\vecpi) = |T|$. Since each outgoing edge from $v$ has capacity $3$ and $|T|=3k$, we know by the definition of edge-ranking games that $|\{w \in U \mid u_w(\vecpi) = 3\}| = k$. We denote these vertices by $w_1, \dots, w_k$. Thus, $u_{w_i}(\vecpi) = 3$ for $i\leq k$ and $u_{w_i}(\vecpi) = 0$ for all $i>k$. We will show that the sets corresponding to vertices $w_1, \dots, w_k$ form a solution to $I$. Suppose there are two sets $w_1, w_2$ that have a non-empty intersection. If this is the case, there is an element $t \in T$ with $t \in w_1 \cap w_2$. There are edges $(w_1,t), (w_2,t)$ in the edge-ranking game that carry flow. This is a contradiction to the fact that $u_v(\vecpi) = |T|$.

Now, let us assume there is a solution $w_1, \dots, w_k$ to $I$. We will show that every pure Nash equilibrium $\vecpi$ in the edge-ranking game yields utility $u_v(\vecpi) = 3k$. The total capacity of all incoming edges of $v$ is exactly $3k$, so $u_v(\vecpi) \leq 3k$ clearly holds for all strategy profiles. We will now argue that independent of the strategy choices of all other players, node $v$ can always obtain $u_v(\vecpi) = 3k$. Let $\pi_v = \left((v,w_1), \dots, (v,w_k), \dots \right)$. Since $w_1, \dots, w_k$ exactly cover all elements $t \in T$, this induces $3k$ cycles in the \TCI\ algorithm discussed in Section~\ref{sec:topCycle}. This is independent of the strategy choices of all other nodes since they always have the property that either all outgoing edges are fully saturated, or there is no flow at all. 

This shows that even in a class of games with guaranteed existence, computing a pure Nash equilibrium or a strong equilibrium is strongly \classNP-hard.

\paragraph{Hardness of computing social optima.}
Consider the previous construction. Every simple cycle in the network involves $v$ and exactly two other nodes. There are no fixed supplies. As such, the social welfare in the system is exactly $3u_v(\vecpi)$. Hence, by computing a best-response for $v$, we also obtain a strategy profile with maximum social welfare. This proves strong \classNP-hardness for the computation of optimal strategy profiles. Note that this result is in contrast to unit-ranking games, where an optimal strategy profile can be computed in strongly polynomial time.

\paragraph{Hardness of deciding existence of Nash and strong equilibria.}
For hardness of the decision version, we adjust the construction from the first part of this proof and combine it with the game without pure Nash equilibrium in Fig.~\ref{fig:noNash}. Observe that fixed supplies of $1$ for node $v_7$ in Fig.~\ref{fig:noNash} would lead to the dominant strategy $\pi_{v_3} = \left((v_3,v_1),(v_3,v_7)\right)$ for $v_3$. Given this strategy, it is straightforward to derive the utilities for all strategy choices of $v_1$ and $v_2$.
\begin{center}
	\begin{tabular}{l || lcr | lcr}
		\diagbox[innerwidth=3cm,innerleftsep=.7cm,innerrightsep=.7cm]{$\pi_{v_1}$}{$\pi\textsl{}_{v_2}$} & \multicolumn{3}{c|}{$((v_2,v_1),(v_2,v_5))$} & \multicolumn{3}{c}{$((v_2,v_5),(v_2,v_1))$} \\\hline\hline
		&   & & 4 &   & & 4\\
		$((v_1,v_4),(v_1,v_6))$ & & & & & & \\
		& 9 & &   & 3 & &  \\\hline
		&   & & 4 &   & & 3\\
		$((v_1,v_6),(v_1,v_4))$  & & & & & & \\ 
		& 9 & &   & 5 & &  \\
	\end{tabular}
\end{center}
Thus, given that $v_7$ has fixed supply of 1 and $v_3$ plays the dominant strategy $\pi_{v_3} = \left((v_3,v_1),(v_3,v_7)\right)$, the strategy profile with $\pi_{v_1} = \left((v_1,v_6),(v_1,v_4)\right)$ and $\pi_{v_2} = \left((v_2,v_1),(v_2,v_5)\right)$ is a strong equilibrium.

We use this insight to design a class of edge-ranking games based on instances of \ThreeDM\ with the following property. If instance $I$ has a solution, the game has a strong equilibrium (and hence a pure Nash equilibrium); if $I$ has no solution, the game has no pure Nash equilibrium (and hence no strong equilibrium). We adjust the construction from the first part of this proof and combine it with $|T|$ copies of the example of Fig.~\ref{fig:noNash}: In the construction from the first part of the proof above, delete all edges $(t,v)$ for all $t \in T$ and add a fixed supply of $b_v = 3k$ for node $v$. We know that each $t \in T$ receives an inflow of 1 if and only if there is a solution to the instance $I$ of \ThreeDM. The instance depicted in Fig.~\ref{fig:noNash} is copied $|T|$ times. We denote the $|T|$ copies of node $v_7$ by $v_7^1, v_7^2, \dots, v_7^{|T|}$. Add an edge with capacity $1$ from each $t_i \in T$ to the corresponding copy of $v_7$, i.e., edges $(t_i, v_7^i)$ for $i=1,\ldots,|T|$. If there is a solution to $I$, all nodes $t_i \in T$ receive a flow of $1$. This flow is forwarded to nodes $v_7^i$ and can be seen as their fixed supply. Hence, there is a strong equilibrium in all copies and also in the game as a whole. On the other hand, if there is no solution to $I$, there is always some $t_i \in |T|$ with $u_{t_i}(\vecpi) = 0$, say, w.l.o.g.\ $u_{t_1}(\vecpi) = 0$. The corresponding node $v_7^1$ does not receive any inflow from the remaining network. Hence, there is no Nash equilibrium in the respective copy and no Nash equilibrium in the game.

\paragraph{Hardness of deciding the existence of a best response with revenue at least $k$.}
The instance constructed in the proof of Theorem~\ref{thm:bestRespNPC} is an edge-ranking game and the proof immediately carries over.
\end{proof}

\begin{remark}\rm
It is unclear whether the problem of deciding existence of a pure Nash equilibrium in an edge-ranking game is in \classNP\ or not, due to \classNP-hardness of verification that a node plays a best response (see Theorem~\ref{thm:bestRespNPC}). It is easy to see that the decision problem is in $\Sigma_2^p$. The problem could be $\Sigma_2^p$-complete, similar to related decision problems in strategic max-flow games~\citep{KupfermanVV17,GuhaKV19}. Proving such a result is an interesting open problem.
\end{remark}

\subsection{Total flow in equilibrium.}

For edge-ranking games, the lower bound on the price of anarchy for Nash equilibria observed in Proposition~\ref{prop:poaUnbounded} applies, i.e., the price of anarchy for Nash equilibria in terms of \rev\ can be unbounded. The restriction to edge-ranking strategies can have a drastic effect even on the total flow in the best equilibrium in case it exists. In particular, in edge-ranking games the price of stability for strong equilibria in terms of \rev\ can be as high as $\Omega(n)$, and the price of stability for Nash equilibria might even be unbounded.

\begin{proposition}
For every $\varepsilon > 0$, there is an edge-ranking game with price of stability for strong equilibria in terms of \rev\ of at least $n/2 - \varepsilon$.
\label{prop:poa}
\end{proposition}

\begin{proof}
We construct an edge-ranking game that consists of a single cycle plus one additional edge. More formally, we have firms $V=\{v_1, \dots, v_n\}$ and edges $E=\Big\{\big\{(v_i,v_{i+1}) \mid i \in \{1,\dots,n-1\}\big\} \cup (v_n,v_1) \cup (v_1,v_n)\Big\}$. The edges $(v_1,v_n), (v_n,v_1),(v_1,v_2)$ have a capacity $M+1$ and all other edges a capacity of $M$. The only node with more than a single outgoing edge is $v_1$. If $\pi_{v_1}=((v_1,v_n),(v_1,v_2))$, player $v_1$ gets a total supply of $M+1$, which is optimal. Observe that $\rev(\vecpi) = 2M+2$, and that $\vecpi$ is the only Nash equilibrium and the only strong equilibrium.

In contrast, for profile $\vecpi'$ with $\pi'_{v_1}=((v_1,v_2),(v_1,v_n))$, firm $v_1$ only gets a utility of $M$, but $\rev(\vecpi') = nM$. Thus, the strong price of stability is at least $nM/(2M+2) = n/2 - n/(2M + 2)$, which is at least $n/2 - \varepsilon$ for $M \ge n/(2\varepsilon) - 1$.
\end{proof}

\begin{proposition}
The price of stability for Nash equilibria in terms of \rev\ is unbounded in edge-ranking games.
\label{prop:posExtrnalInflow}
\end{proposition}

\begin{figure}[t]
\begin{center}
\begin{tikzpicture}[auto,swap,scale=0.745]
\node[vertex](v1) at (0,0){$v_1$};
\node[vertex](v4) at (-2,-2){$v_4$};
\node[vertex,label=above:\fbox{2}](v2) at (-4,0){$v_2$};
\node[vertex](v6) at (-6,0){$v_5$};
\node[vertex](v7) at (2,-2){$v_6$};
\node[vertex,label=above:\fbox{2}](v3) at (4,0){$v_3$};
\node[vertex](v9) at (6,0){$v_7$};
\node[vertex,label=above:\fbox{1}] (w1) at (8,1.5) {$w_1$};
\node[vertex] (w2) at (8,-1.5) {$w_2$};
\node[vertex] (w3) at (10,0) {$w_3$};
\path[edge] (v1) edge [above left] node{$4$} (v4);
\path[edge] (v4) edge [left] node {$2$}(v2);
\path[edge] (v2) edge [above] node{$6$}(v1);
\path[edge] (v2) edge [bend left = 20, below] node{$6$} (v6);
\path[edge] (v6) edge [bend left = 20, above] node{$1$} (v2);
\path[edge] (v1) edge node[above right] {$4$} (v7);
\path[edge] (v7) edge [right] node {$2$}(v3);
\path[edge] (v3) edge [above] node{$6$}(v1);
\path[edge] (v3) edge [bend right = 20, below] node{$6$} (v9);
\path[edge] (v9) edge [bend right = 20, above] node{$1$} (v3);
\path[edge] (w1) edge [right] node{$M$} (w2);
\path[edge] (w2) edge [below right] node{$M$} (w3);
\path[edge] (w3) edge [above right] node{$M$-2} (w1);

\path[edge] (w1) edge [above left] node{$2$} (v9);
\path[edge] (w2) edge [below left] node{$2$} (v9);

\end{tikzpicture}
\end{center}
\caption{An edge-ranking game with unbounded price of stability for Nash equilibria.}
\label{fig:posExternalInflow}
\end{figure}%

\begin{proof}
Consider the game in Fig.~\ref{fig:posExternalInflow}, which extends the game without pure equilibrium from Fig.~\ref{fig:noNash}. We add three nodes. $w_1$ has fixed supply $1$, $w_2$ and $w_3$ no fixed supply. These nodes are involved in a cycle $C$ of edges $(w_1, w_2)$ and $(w_2,w_3)$ with capacity $M \gg 2$, as well as edge $(w_3,w_1)$ with capacity $M-2$. In addition, there are edges $(w_1,v_9)$ and $(w_2,v_9)$ of capacity 2.

In an optimal circulation, $w_1$ and $w_2$ prioritize the edges of $C$, leading to utilities of $\Theta(M)$. In contrast, a pure Nash equilibrium can only exist if the $w$-nodes ensure that the fixed supply of $w_1$ is forwarded to $v_7$, in which case a Nash equilibrium can exist (as observed in the proof of Theorem~\ref{thm:edgeExistNPC}). Clearly, both $w_1$ and $w_2$ have an incentive to deviate towards $C$. Hence, if \emph{either $w_1$ or $w_2$} places the edge to $v_7$ in first rank and the other does not, a unilateral deviation suffices to close $C$ -- thereby leaving the $v$-nodes with instability. However, if \emph{both $w_1$ and $w_2$} play strategies $\pi_{w_1} = ((w_1,v_7),(w_1,w_2))$ and $\pi_{w_2} = ((w_2,v_7),(w_2,w_3))$, no unilateral deviation can lead to flow along $C$. In this case, a pure Nash equilibrium evolves. Obviously the total flow in this equilibrium is at most a constant. Hence, the price of stability is as large as $\Omega(M)$.
\end{proof}

\section{Conclusions.}
\label{sec:conclude}

In this paper, we have proposed and analyzed flow allocation games. Our main results show that in these games, if firms are following priority rankings over units of flow (i.e., unit-ranking strategies), there is always a strong equilibrium. Moreover, it can be computed in strongly polynomial time and represented in polynomial space. More generally, for a large class of \emph{flow-monotone} quality functions $\sw$, there is even an \emph{optimal} strong equilibrium, i.e., the price of stability for Nash and strong equilibria in terms of $\sw$ is 1. In terms of computational complexity, the properties of such optimal strong equilibria depend highly on the quality function $\sw$. While for some objectives such as \rev\ the optimization problem of finding an optimal strong equilibrium can be solved in polynomial time, for other objectives such as \textsc{FSN} the optimization problem can become strongly \classNP-hard. Alternatively, when restricting the strategy spaces to priorities over edges (i.e., edge-ranking strategies), pure Nash and strong equilibria can be absent, and even deciding their existence is a \classNP-hard problem. 

As a concrete example of a quality function, our results shed further light on the performance of equilibria in terms of \rev. When considering decentralized clearing and arbitrary strong equilibria in unit-ranking games, the price of anarchy for strong equilibria depends on the length of cycles in the money circulation of an optimum profile. For pure Nash equilibria, the deterioration in \rev\ can be severe due to the lack of coordination among firms. For edge-ranking strategies, pure Nash and strong equilibria can have very poor quality in terms of \rev\ when they exist.

Our work provides a game-theoretic perspective on clearing in financial networks. In the context of financial networks, unit-ranking functions allow a centralized market regulator to obtain a clearing state, in which a desired social quality is maximized and no coalition of firms gets an incentive to deviate. As such, our work reveals an intersting alignment of incentives -- firms (even \emph{groups} of firms) share an intrinsic interest to implement these proposed clearing payments. This \emph{strategic robustness} represents an elegant game-theoretic complement to standard \emph{legal enforcement} by regulators in financial networks.

More generally, while we have initiated the analysis of flow allocation games, numerous recent follow-up works in the context of financial networks (c.f.\ our discussion in Section~\ref{sec:related}) show that our work has sparked interest in a variety of directions. Understanding these issues continues to be of vital importance to improve the financial system and to inform the discussion about financial regulation from a computational and game-theoretic perspective.

\section*{Acknowledgments.}
The authors thank Pascal Lenzner and Steffen Schuldenzucker for valuable discussions and feedback on the results of this paper. The authors acknowledge financial support by DFG Research Group ADYN under grant DFG 411362735. NB thanks Dr.~h.~c.~Maucher for funding his position. 

\bibliographystyle{plainnat}
\bibliography{literature}

\end{document}